%% file: CSTheory_stdFormat1.tex
\pdfoutput=0

\documentclass[11pt]{article}
\def\comment#1{}
\usepackage{graphicx} 
\usepackage[all]{xy} 
\usepackage{amsfonts, amsmath, amssymb, url} 
\usepackage{amsthm}
\usepackage{epsfig}
\usepackage{ifthen}
\usepackage{epstopdf}

\def\bbR{\mathbb R}

\def\P{\mathbb P}
\def\what{\widehat}
\def\wtilde{\widetilde}
\def\eqref#1{(\ref{#1})}
\def\lb{\langle}
\def\rb{\rangle}
\def\Q{{\mathbb Q}}
\newtheorem{Thm}{Theorem}
\newtheorem{Lem}{Lemma}
\newtheorem{Prop}{Proposition}
\newtheorem{Rem}{Remark}
\newtheorem{Example}{Example}

\include{include_macros}

\newboolean{showFigure}
\setboolean{showFigure}{false} 

\begin{document}\sloppy
\title{Conditional sampling for spectrally discrete max-stable random fields\thanks{
This work was partially supported by the NSF grant DMS--0806094 at the University of  Michigan, Ann Arbor.}}

\author{Yizao Wang and Stilian A. Stoev
\thanks{{\it Address:}
 Department of Statistics, The University of Michigan, 439 W.\ Hall, 1085 S.\ University, Ann Arbor,
  MI 48109--1107. U.S.A.
  {E--mails:} 
  \texttt{\{yizwang, sstoev\}@umich.edu}.
  }} 
\maketitle

\comment{
\authornames{Y.\ Wang and S.\ Stoev} 
\shorttitle{Conditional sampling for max-stable random fields} 
\title{Conditional sampling for spectrally discrete\\ max-stable random fields}
\authorone[University of Michigan, Ann Arbor, U.S.A.]{Yizao Wang}
\authorone[University of Michigan, Ann Arbor, U.S.A.]{Stilian A.\ Stoev}
\addressone{Department of Statistics, the University of Michigan, 439 West Hall, 1085 South University, Ann Arbor,
  MI 48109--1107. {\it E--mails:} {\tt \{yizwang, sstoev\}@umich.edu} } 
}
\begin{abstract}
Max-stable random fields play a central role in modeling extreme value phenomena.  
We obtain an explicit formula for the conditional probability in general max-linear models,
which include a large class of max-stable random fields. As a consequence, we develop an algorithm 
for efficient and exact sampling from the conditional distributions. Our method provides a computational 
solution to the prediction problem for spectrally discrete max-stable random fields. 
This work offers new tools and a new perspective to many statistical inference problems for spatial
extremes, arising, for example, in meteorology, geology, and environmental applications.
\end{abstract}

\section{Introduction}

 \subsection{Motivation}
Max-stable stochastic processes and random fields are fundamental statistical models for 
the dependence of extremes. This is because they arise in the limit of rescaled maxima. 
Indeed, consider the component-wise maxima 
$$
 M_t^{(n)} = \max_{j = 1,\dots,n} \xi_t^{(j)},\ \ t\in T
$$
of independent realizations $\{\xi_t^{(j)}\}_{t\in T}$, $j=1,\dots,n$ of a random field $\xi=\{\xi_t\}_{t\in T}$. 
If the random field $\{M_t^{(n)}\}_{t\in T}$ converges in law, as $n\to\infty$, under judicious normalization, then its limit $X = \indt X$ is necessarily {\em max-stable} (see e.g.~Resnick~\cite{resnick87extreme} and de Haan and Ferreira~\cite{dehaan06extreme}).

Therefore, the max-stable processes (random fields, resp.) are as important to extreme value theory as are the Gaussian processes 
to the classical statistical theory based on the central limit theorem. The multivariate max-stable laws and processes have been studied extensively 
in the past 30 years. See e.g.~Balkema and Resnick~\cite{balkema77max}, de Haan~\cite{dehaan78characterization,dehaan84spectral}, de Haan and Pickands~\cite{dehaan86stationary}, 
Gin\'e et al.~\cite{gine90max}, Smith~\cite{smith90max}, Resnick and Roy~\cite{resnick91random}, Davis and Resnick~\cite{davis89basic,davis93prediction}, Stoev and Taqqu~\cite{stoev06extremal}, Kabluchko et al.~\cite{kabluchko09stationary}, 
Wang and Stoev~\cite{wang10structure}, among many others.

The modeling and parameter estimation of the univariate marginal distributions of the extremes have been studied extensively (see 
e.g. Davison and Smith~\cite{davison90models}, de Haan and Ferreira~\cite{dehaan06extreme}, Resnick~\cite{resnick07heavy} 
and the references therein). Many of the recent developments in this domain focus on 
the characterization, modeling and estimation of the dependence for multivariate extremes. In this context, building adequate max-stable processes and random fields plays a key role. See e.g.~de Haan and Pereira~\cite{dehaan06spatial}, Buishand et al.~\cite{buishand08spatial},
Schlather~\cite{schlather02models}, Schlather and Tawn~\cite{schlather03dependence}, Cooley et al.~\cite{cooley07bayesian}, and Naveau et al.~\cite{naveau09modelling}.

Our present work is motivated by an important and long-standing challenge, namely,
the prediction for max-stable random processes and fields. Suppose that
one already has a suitable max-stable model for the dependence structure of a
random field $\indt X$.  The field is observed at several
locations $t_1,\dots, t_n \in T$ and one wants to {\em predict} the
values of the field $X_{s_1},\dots,X_{s_m}$ at some other locations.
The optimal predictors involve the {\em conditional distribution} of
$\indt X$, given the data.  Even if the finite-dimensional distributions of
the field $\indt X$ are available in analytic form, it is typically
impossible to obtain a closed-form solution for the conditional
distribution.  Na\"ive Monte Carlo approximations are not practical either,
since they involve conditioning on events of infinitesimal
probability, which leads to mounting errors and computational costs.

Prior studies of Davis and Resnick~\cite{davis89basic,davis93prediction} and Cooley et al.~\cite{cooley07bayesian}, among others, have shown that the prediction problem
in the max-stable context is challenging, and it does not have an elegant analytical solution. On the other hand, the growing popularity and the use of 
max-stable processes in various applications, make this an important problem. This motivated
us to seek a {\em computational} solution.

In this work, we develop theory and methodology for sampling from the conditional distributions of spectrally discrete
max-stable models.  More precisely, we provide an algorithm that can generate efficiently {\em exact} independent samples from the 
{\em regular conditional probability} of  $(X_{s_1},\dots,X_{s_m})$, given the values $(X_{t_1},\dots,X_{t_n})$. For the sake of simplicity, we write $\vv X = (X_1,\dots,X_n) \equiv (X_{t_1},\dots,X_{t_n})$.
The algorithm applies to the general {\it max-linear model}:
\equh\label{rep:maxLinear_old}
 X_{i} = \max_{j = 1,\dots,p}a_{i,j}Z_j\equiv \bveejp a_{i,j}Z_j\,, i = 1,\dots,n.
\eque
where the $a_{i,j}$'s are known non-negative constants and the $Z_j$'s are independent continuous 
non-negative random variables. Any multivariate max-stable distribution can be approximated arbitrarily well
via a max-linear model with sufficiently large $p$.
 
The main idea is to first generate samples from the {\em regular conditional probability} distribution of
$\vv Z \mid \vv X = \vv x$, where $\vv Z = (Z_j)_{j = 1,\dots, p}$. Then, the conditional distributions of
$$
 X_{s_k} = \bveejp b_{k,j}Z_j\,, 1\leq k\leq m,
$$
given $\vv X = \vv x$ can be readily obtained, for any given $b_{k,j}$'s. In this paper, we assume that the model is completely known, i.e., the parameters $\{a_{i,j}\}$ and $\{b_{k,j}\}$ are given. The statistical inference for these parameters is beyond the scope of this paper.

Observe that if $\vv X = \vv x$, then~\eqref{rep:maxLinear_old} implies natural equality and inequality constraints on the $Z_j$'s. More precisely,~\eqref{rep:maxLinear_old} gives rise to a set of so-called {\em hitting scenarios}.
In each hitting scenario, a subset of the $Z_j$'s equal, in other words {\em hit}, their upper bounds and the rest of the $Z_j$'s can take arbitrary values in certain open intervals. We will show that the regular conditional probability of $\vv Z\mid\vv X = \vv x$ is a weighted mixture of the various distributions of the vector $\vv Z$, under all possible hitting scenarios corresponding to $\vv X = \vv x$.

The resulting formula, however, involves determining {\em all} hitting scenarios, which becomes computationally prohibitive for large and even moderate values of $p$. This issue is closely related to the NP-hard {\em set-covering problem} in computer science (see e.g.~\cite{caprara00algorithms}).

Fortunately, further detailed analysis of the probabilistic structure of the max-linear models allows us to obtain a different formula of the regular conditional probability (Theorem~\ref{thm:factorization}). It yields an exact and computationally efficient algorithm, which in practice can handle complex max-linear models with $p$ in the order of thousands, on a conventional desktop computer. The algorithm is implemented in the R (\cite{R09R}) package {\tt maxLinear}~\cite{wang10maxLinear}, with the core part written in C/C++. We also used the R package {\tt fields} (\cite{furrer09fields}) to generate some of the figures in this paper.

We illustrate the performance of our algorithm over two classes of processes: the max-autoregressive moving average (MARMA) time series (Davis and Resnick~\cite{davis89basic}), and the Smith model (Smith~\cite{smith90max}) for spatial extremes. The MARMA processes are spectrally discrete max-stable processes, and our algorithm applies directly. In Section~\ref{sec:MARMA}, we demonstrate the prediction of MARMA processes by conditional sampling and compare our result to the projection predictors proposed in~\cite{davis89basic}. To apply our algorithm to the Smith model, on the other hand, we first need to discretize the (spectrally continuous) model. Section~\ref{sec:Smith} is devoted to conditional sampling for the discretized Smith model. Thanks to the computational efficiency of our algorithm, we can choose a mesh fine enough to obtain a satisfactory discretization. Figure~\ref{fig:samplings} shows four realizations from such a discretized Smith model, conditioning only on 7 observations (with assumed value 5). 
The algorithm applies in the same way to more complex models.  
\begin{figure}[ht!]
  \begin{center}
  \ifthenelse{\boolean{showFigure}}{
     \includegraphics[width = .8\textwidth]{CS_4samples.pdf}
   }{\includegraphics[width = .8\textwidth]{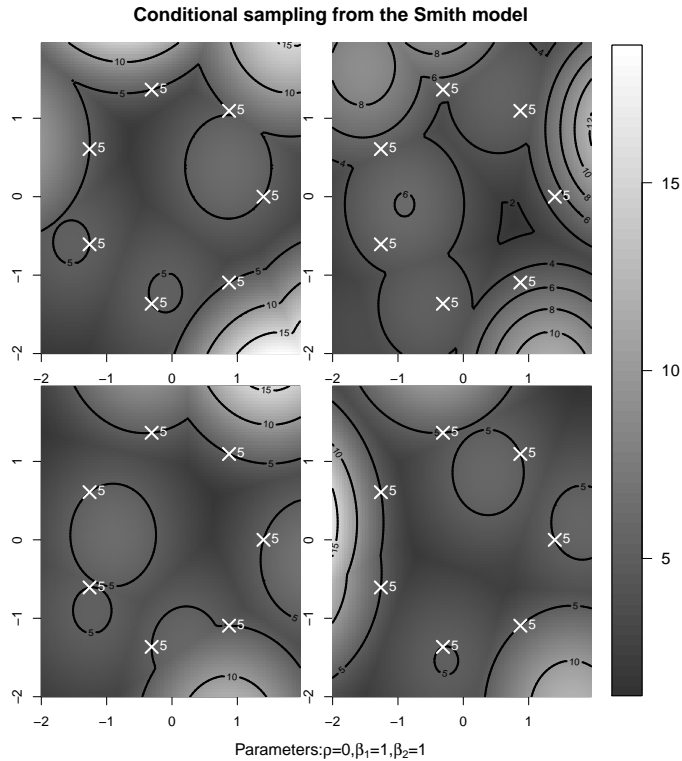}}
   \caption{\label{fig:samplings} Four samples from the conditional distribution of the discrete 
   Smith model (see Section~\ref{sec:Smith}), given the {\em observed} values (all equal to 5) at the locations marked by crosses.}
  \end{center}
 \end{figure}


\subsection{Multivariate Max-Stable Distributions: a Brief Review}\label{sec:approximation}

Consider a general max-stable process $X = \indt X$, indexed by a set
$T$ (e.g.\ $T = [0,1], \mathbb R, \mathbb R^d$ or $\mathbb Z^d$). We
shall assume that the finite-dimensional distributions of $X$ are
known and the ultimate goal is to study the conditional distributions
of $X$.  For convenience and without loss of
generality, we focus on max-stable processes $X$ with
$\alpha$-Fr\'echet marginals ($\alpha>0$), such that all
max-linear combinations
$$
 \xi =\max_{j=1,\dots,n} a_j X_{t_j} \equiv \bveein a_j X_{t_j},\ \ a_j>0,\ t_j\in T,
$$ 
have the $\alpha$-Fr\'echet distribution:
$$
\P(\xi\le x) = \exp\{ -\sigma_\xi^\alpha x^{-\alpha}\},\ \ x\in (0,\infty),
$$
with scale coefficient $\sigma_\xi>0$.  Any max-stable process can be related to such an $\alpha$-Fr\'echet process 
by simple transformation of the marginals (see e.g.~\cite{resnick87extreme}).

Essentially all max-stable processes $\indt X$ admit the following {\it extremal integral representation}:
\equh\label{rep:maxStable}
 \indt X\eqd \bccbb{\Eintt_S f_t(s)M_\alpha(\d s)}_{t\in T}\,,
\eque
where the $f_t$'s are non-negative, measurable deterministic functions defined on a suitable measure space $(S,\mu)$ and such that
$\int_S f_t^\alpha(s)\mu(\d s)<\infty$. Here $M_\alpha$ is an {\it $\alpha$-Fr\'echet random sup-measure} with control 
measure $\mu$.  The extremal integral $\eintt_S f\d M_\alpha$ can be defined for all $f\in L^\alpha(S,\mu),\ f\ge 0$,
as the limit in probability of extremal integrals of simple functions. For more details, see~\cite{stoev06extremal} and the seminal 
work \cite{dehaan84spectral} for an alternative treatment. 

The functions $\indt f$ are called the {\it spectral functions} of the process $\indt X$. They determine the finite-dimensional distributions
of $\indt X$:
\[
\proba(X_{t_1}\leq x_1,\cdots, X_{t_n}\leq x_n) = \exp\bccbb{-\int_S\bpp{\bveein f_{t_i}(s)/x_i}^\alpha\mu(\d s)}\,,
\]
for all $t_i\in T,\ x_i\in \mathbb R_+:=(0,\infty),\ i = 1,\dots, n$. A popular equivalent representation of multivariate
max-stable laws is as follows:
$$
 \proba(X_{t_1}\leq x_1,\cdots, X_{t_n}\leq x_n) = \exp{\Big\{} - \int_{{\mathbb S}_+^{n-1} } {\Big(} \bveein w_i/x_i {\Big)}^\alpha 
\Gamma(\d \vv w){\Big\}}.
$$
Here ${\mathbb S}_+^{n-1} = \{ \vv w = (w_i)_{i=1}^n \in \mathbb R^n\, :\, 0\le w_j \le \max_{i = 1,\dots, n} w_i =1\}$ is the positive 
unit sphere in the sup-norm, and $\Gamma$ is a unique finite measure on $ {\mathbb S}_+^{n-1}$ called the {\em spectral measure}
of the distribution (see e.g.~\cite{resnick87extreme,dehaan06extreme}).

Any multivariate max-stable vector $(X_{t_j})_{j=1}^n$ can be approximated arbitrarily well in probability, by discretizing the extremal integral:
\[
X_{t_i} = \Eintt_S f_{t_i}(s)M_\alpha(\d s) \approx \bigvee_{i=1}^p a_{i,j} Z_j,
\]
where $Z_j,\ j = 1,\dots, p$ are independent standard $\alpha$-Fr\'echet variables and $a_{i,j}\ge 0$. This is equivalent to considering multivariate 
max-stable vectors with discrete spectral measures concentrated on at most $p$ points on the unit sphere ${\mathbb S}_+^{n-1}$. 
The error of approximation, moreover, can be controlled explicitly through convenient probability metrics (see e.g.~\cite{stoev06extremal}).

In this paper, we shall focus on the class of max-stable processes:
$$
 X_t \defe \bigvee_{j=1}^p \phi_j(t) Z_j,\ t\in T,
$$
where the $\phi_j(t)$'s are non-negative deterministic functions. These processes are called {\em spectrally discrete},
since their spectral measures $\Gamma$ are discrete. By taking sufficiently large $p$'s and with judicious $\phi_j(t)$'s, one can build flexible
models that can replicate the behavior of an arbitrary  max-stable process.  From this point of view, a satisfactory computational solution must be able to deal with max-linear models with large $p$'s.

The treatment of the exact conditional distributions of general spectrally continuous max-stable processes requires different tools and still remains
an open problem, to the best of our knowledge. As we shall see, the solution in the discrete case, although complete, is already quite involved.

{\bf Acknowledgments.} 
The authors were partially supported by NSF grant DMS--0806094 at the University of Michigan.


\section{Conditional Probability in Max-Linear Models}\label{sec:conditionalProbability}

 \subsection{Intuition and Basic Theory}\label{sec:formula1}

Consider the max-linear model in~\eqref{rep:maxLinear_old}. We shall denote this model by:
\begin{equation}\label{rep:maxLinear}
 \vv X = A\odot\vv Z,
\end{equation}
where $A = (a_{i,j})_{n\times p}$ is a matrix with non-negative entries, $\vv X = (X_1,\dots, X_n)$ and $\vv Z = (Z_1,\dots,Z_p)$ are 
column vectors. We assume that the $Z_j$'s, $j = 1,\dots, p$, are independent non-negative random variables
having probability densities. 

In this section, we provide an explicit formula for the {\em regular conditional probability} of $\vv Z$ with respect to 
$\vv X$ (see Theorem \ref{thm:rcp} below and the Appendix for a precise definition). We start with some intuition and notation.
Throughout this paper, we assume that the matrix $A$ has at least one nonzero entry
in each of its rows and columns. This will be referred to as {\em Assumption A}.

Observe that if $\vv x = A\odot \vv z$ with $\vv x \in \bbR_+^n,\ \vv z\in \bbR_+^p$, then
\begin{equation}\label{eq:zhat}
  0\leq z_j\leq \widehat z_j \equiv  \widehat z_j (A,\vv x)\defe \min_{1\le i \le n} x_i/a_{i,j},\ \ \ j = 1,\dots, p.
\end{equation}
That is, the max-linear model \eqref{rep:maxLinear} imposes certain inequality and equality constraints on the $Z_j$'s, 
given a set of {\em observed} $X_i$'s.  Namely, some of the upper bounds $\widehat z_j(A,\vv x)$ in \eqref{eq:zhat} must be attained, or {\em hit}, i.e., $z_j = \widehat z_j(A,\vv x)$ in such a way that 
 $$
  x_i = a_{i,j(i)} z_{j(i)},\ \ \ i = 1,\dots, n,
 $$
 with judicious $j(i)\in \{1,\dots,p\}$.  The next example helps to understand the inequality and equality constraints.
 \begin{Example}\label{example:0}
 Suppose that $n = p = 3$ and
$$
A = 
\left(\begin{array}{ccc}
1 & 0 & 0\\
1 & 1 & 0\\
1 & 1 & 1
\end{array}\right)\,.
$$
Let $\vv x = A\odot \vv z$ for some $\vv z\in\mathbb R_+^3$. In this case, it necessarily follows that $x_1\leq x_2\leq x_3$. Moreover,~\eqref{eq:zhat} yields $\what {\vv z} = \vv x$.
\begin{itemize}
\item [(i)] If $\vv x = (1,2,3)$, then it trivially follows that $\vv z= \what{\vv z} = (1,2,3)$, which is an equality constraint on $\vv z$.  
\item [(ii)] If $\vv x = (1,1,3)$, then it follows that $z_1 = \what z_1 = 1, z_2\leq \what z_2 = 1$ and $z_3 = \what z_3 = 3$. Here, the ``equality constraints'' must hold for $z_1 = \what z_1$ and $z_3 = \what z_3$, while $z_2$ only needs to satisfy the ``inequality constraint'' $0\leq z_2\leq \what z_2$. 
\end{itemize}
\end{Example}
Write
 $$
  C(A,\vv x) \defe \{ \vv z\in \mathbb R_+^p\, :\, \vv x = A\odot \vv z\},
 $$
 and note that the conditional distribution of $\vv Z\mid\vv X = \vv x$ concentrates on the set $C(A,\vv x)$. The observation in Example~\ref{example:0} can be generalized and formulated as follows. 
\begin{itemize}
\item Every $\vv z \in C(A,\vv x)$ corresponds to a set of {\em active} (equality) constraints $J \subset \{1,\dots,p\}$, 
 which we refer to as a {\it hitting scenario} of $(A,\vv x)$, such that
 \equh\label{eq:scenario}
 z_j = \widehat z_j(A,\vv x),\ \ j\in {J}\mand z_j < \widehat z_j(A,\vv x),\ \  j\in 
 {J}^{c}\defe\{1,\dots,p\}\setminus J.
 \eque
 Observe that if $j\not \in J$, then there are no further constraints and $z_j$ can take any value in $[0,\widehat z_j)$, regardless
 of the values of the other components of the vector $\vv z \in C(A,\vv x)$.
 \item Every value $\vv x$ may give rise to many different hitting scenarios $J \subset\{1,\dots,p\}$.
 Let $\calJ(A,\vv x)$ denote the collection of all such $J$'s. We refer to $\calJ(A,\vv x)$ as to the {\em hitting distribution} of $\vv x$ 
 w.r.t.\ $A$:
 \[
 \calJ(A,\vv x) \equiv \bccbb{J\subset\{1,\dots,p\}: \mbox{ exist } \vv z\in C(A,\vv x), \mbox{such that\ } \eqref{eq:scenario} \mbox { holds }}.
 \]
\end{itemize}
To illustrate the notions of hitting scenario and hitting distribution, consider again Example~\ref{example:0}. Therein, we have $\calJ(A,\vv x) = \{\{1,2,3\}\}$ in case (i), and $\calJ(A,\vv x) = \{\{1,3\},\{1,2,3\}\}$ in case (ii).

The hitting distribution $\calJ(A,\vv x)$ is a finite set and thus can always be identified. However, the identification procedure is the key difficulty in providing an efficient algorithm for conditional sampling in practice. This issue is addressed in Section~\ref{sec:algorithm}.
In the rest of this section, suppose that $\calJ(A,\vv x)$
is given. Then, we can partition $C(A,\vv x)$ as follows
 \[
  C(A,\vv x) = \bigcup_{J \in {\cal J}(A,\vv x)} C_J(A,\vv x)\,,
  \]
  where 
  \[
  C_J(A,\vv x)=\{ \vv z \in \mathbb R_+^p\, :\, z_j = \widehat z_j,\ j\in J\mbox{ and }z_j <\widehat z_j,\ j\not\in J\}.
 \]
 The sets $C_J(A,\vv x),\ J\in {\cal J}(A,\vv x)$ are disjoint since they correspond to different {\em hitting scenarios} 
 in ${\cal J}(A,\vv x)$. Let 
 \equh\label{eq:rJ}
 r(\calJ(A,\vv x)) = \min_{J\in\calJ(A,\vv x)}|J|\,,
 \eque
 where $|J|$ is the number of elements in $J$. We call $r(\calJ(A,\vv x))$ the {\em rank} of the hitting distribution $\calJ(A,\vv x)$. It equals
 the minimal number of equality constraints among the hitting scenarios in $\calJ(A,\vv x)$. It will turn out that the hitting 
 scenarios $J \subset \calJ(A,\vv x)$ with $|J| > r(\calJ(A,\vv x))$ occur with (conditional) probability zero and can be ignored. 
 We therefore focus on the set of all {\em relevant} hitting scenarios:
 $$
 {\cal J}_r(A,\vv x) = \{J\in \calJ(A,\vv x): |J| = r(\calJ(A,\vv x))\}.
 $$
\begin{Thm}\label{thm:rcp} Consider the max-linear model in \eqref{rep:maxLinear},  where
$Z_j$'s are independent random variables with densities $f_{Z_j}$ and distribution functions $F_{Z_j}$, $j = 1,\dots, p$. Let
$A = (a_{i,j})_{n\times p}$ have non-negative entries satisfying Assumption A and let $\calR_{\mathbb R_+^p}$ be the class of all rectangles $\{ (\vv e, \vv f],\ \vv e, \vv f\in \mathbb R_+^p\}$ in ${\mathbb R_+^p}$.

For all $J\in {\cal J}(A,\vv x)$, $E\in \calR_{\mathbb R_+^p}$, and $\vv x\in\mathbb R_+^n$, define
\equh\label{eq:nuJ}
 \nu_{J}(\vv x, E) \defe \prod_{j\in J} 
   \delta_{\widehat z_j}(\pi_j(E)) \prod_{j\in J^c}\P\{ Z_j\in \pi_j(E)\mid Z_j < \widehat z_j\},
\eque
where $\pi_j(z_1,\dots,z_p) = z_j$ and $\delta_a$ is a unit point-mass at $a$.

Then, the  {\em regular conditional probability} $\nu(\vv x,E)$ of $\vv Z$ w.r.t.\ $\vv X$ equals:
\begin{equation}\label{eq:nu}
 \nu(\vv x,E) = \sum_{J \in {\calJ_r}(A,\vv x)} p_J(A,\vv x) \nu_J(\vv x,E),\ \  E\in {\calR_{\mathbb R_+^p}},
\end{equation}
for $\proba^{\vv X}$-almost all $\vv x\in A\odot(\mathbb R^p_+)$, where for all $J\in\calJ_r(A,\vv x)$, 
\begin{equation}\label{eq:pJ}
p_J(A,\vv x) = \frac{w_J}{\sum_{K\in {\calJ_r}(A,\vv x)} w_K} \qmwith  w_J
        = \prod_{j\in J} \widehat z_j f_{Z_j}(\widehat z_j) \prod_{j\in J^c} F_{Z_j}(\widehat z_j).
\end{equation}
\end{Thm}
\noindent In the special case when the $Z_j$'s are $\alpha$-Fr\'echet with scale coefficient 1, we have $w_{J} = \prod_{j\in J}({\widehat z_j})^{-\alpha}$.

\begin{Rem}\label{rem:measurability}
We state \eqref{eq:nu} {\em only} for rectangle sets $E$ because the projections $\pi_j(B)$ of an arbitrary 
Borel set $B\subset \mathbb R_+^p$ are not always Borel (see e.g.~\cite{srivastava98course}). 
Nevertheless, 
the extension of measure theorem ensures that Formula~\eqref{eq:nu} specifies completely the regular 
conditional probability.
\end{Rem}
We do not provide a proof of Theorem~\ref{thm:rcp} directly. Instead, we will first provide an equivalent formula for
$\nu(\vv x,E)$ in Theorem~\ref{thm:factorization} in Section~\ref{sec:algorithm}, and then prove that $\nu(\vv x,E)$ is the desired regular
conditional probability. All the proofs are deferred to Section~\ref{sec:proofs}. The next 
example gives the intuition behind Formula \eqref{eq:nu}. 
\begin{Example}\label{Ex:1}
Continue with Example~\ref{example:0}.
\begin{itemize}
\item [(i)]
If $\vv X = \vv x = (1, 2, 3)$, then $\what {\vv z} = \vv x$, $\calJ(A,\vv x) = \{\{1,2,3\}\}$. Therefore, $r(\calJ(A,\vv x)) = 3$ and Formula~\eqref{eq:nu} yields
$$
\nu(\vv x,E) = \nu_J(\vv x,E) = 
 \delta_{\what z_1} (\pi_1(E)) \delta_{\what z_2} (\pi_2(E)) \delta_{\what z_3} (\pi_3(E)) \equiv
  \delta_{\vv {\what z}}(E)\,,
$$
a degenerate distribution with single unit point mass at $\what{\vv z}$.
\item [(ii)] If $\vv X = \vv x = (1, 1, 3)$, then, $\what {\vv z} = \vv x$, $\calJ(A,\vv x) = \{\{1,3\}, \{1,2,3\}\}$, and $r(\calJ(A,\vv x)) = 2$. Therefore, $\calJ_r(A,\vv x) = \{\{1,3\}\}$ and Formula~\eqref{eq:nu} yields:
$$
\nu(\vv x,E) = \nu_{\{1,3\}}(\vv x,E) =  \delta_{\what z_1} (\pi_1(E)) \proba( Z_2 \in \pi_2(E)\mid Z_2 < \what z_2) \delta_{\what z_3} (\pi_3(E)).
$$
In this case, the conditional distribution concentrates on the one-dimensional set $\{1\}\times(0,1)\times\{3\}$.

\item[(iii)] Finally, if $\vv X = \vv x = (1, 1, 1)$, then $\what{\vv z} = \vv x$ and $\calJ(A,\vv x) = \{\{1\}, \{1,2\},\{1,2,3\}\}$. Then, $\calJ_r(A,\vv x) = \{\{1\}\}$ and
\[
\nu(\vv x,E) = \nu_{\{1\}}(\vv x,E) =  \delta_{\what z_1} (\pi_1(E)) \prod_{j=2}^3\proba( Z_j \in \pi_j(E)\mid Z_j < \what z_j).
\]
The conditional distribution concentrates on the set $\{1\}\times (0,1)\times (0,1)$.
\end{itemize}
\end{Example}
We conclude this section by showing that the conditional distributions~\eqref{eq:nu} arise as suitable limits. 
This result can be viewed as a heuristic justification of Theorem~\ref{thm:rcp}. Let $\epsilon>0$, consider 
\equh\label{eq:CJepsilon}
C_J^\epsilon(A,\vv x) \defe \bccbb{\vv z\in\mathbb R^p_+: z_j\in[\widehat z_j(1-\epsilon),\widehat z_j(1 + \epsilon)], j\in J, z_k<\widehat z_k(1-\epsilon)\,,k\in J^c},
\eque
and set
\equh\label{eq:Cunion}
C^\epsilon(A,\vv x) \defe \bigcup_{J\in\calJ(A,\vv x)} C_J^\epsilon(A,\vv x)\,.
\eque
Note that the sets $A\odot (C^\epsilon(A,\vv x))$ shrink to the point $\vv x,$ as $\epsilon\downarrow0$.  

\begin{Prop}\label{prop:limit}
Under the assumptions of Theorem~\ref{thm:rcp}, for
all $\vv x\in A\odot(\mathbb R_+^p)$, we have, as $\epsilon\downarrow 0$,
\equh \label{eq:limit}
  \proba(\vv Z\in E \mid\vv Z\in C^\epsilon(A,\vv x))  \longrightarrow  \nu(\vv x,E),\ 
   E \in {\calR_{\mathbb R_+^p}}.
\eque
\end{Prop}
\begin{proof}
Recall the definition of $C^\epsilon_J$ in~\eqref{eq:CJepsilon}. Observe that for all $\epsilon>0$, the sets $\{C_J^\epsilon(A,\vv x)\}_{J\in\calJ(A,\vv x)}$ are mutually disjoint. Thus, writing $C^\epsilon \equiv C^\epsilon(A,\vv x)$ and $C^\epsilon_J\equiv C^\epsilon_J(A,\vv x)$, by~\eqref{eq:Cunion} we have
\eqnhn
\proba(\vv Z\in E\mid\vv Z\in C^\epsilon) & = & \sum_{J\in\calJ}\proba(\vv Z\in E\mid\vv Z\in C_J^\epsilon)\proba(\vv Z\in C_J^\epsilon\mid\vv Z\in C^\epsilon)\nonumber\\
& = & \sum_{J\in\calJ}\proba(\vv Z\in E\mid\vv Z\in C_J^\epsilon) \frac{\proba(\vv Z\in C_J^\epsilon)}{\sum_{K\in\calJ}\proba(\vv Z\in C_K^\epsilon)}\label{eq:approximation},
\eqnen
where the terms with $\proba(\vv Z\in C_J^\epsilon)=0$ are ignored. One can see that $\proba(\vv Z\in E\mid\vv Z\in C_J^\epsilon)$ converge to $\nu_J(E,\vv x)$ in \eqref{eq:nuJ}, as $\epsilon\downarrow 0$. The independence of the 
$Z_j$'s also implies that
\begin{multline}
\proba(\vv Z\in C_J^\epsilon)= 
 \prod_{j\in J}\proba(Z_j\in[\widehat z_j(1 - \epsilon), \widehat z_j(1 + \epsilon)])\prod_{k\in J^c}\proba(Z_k\leq \widehat z_k (1 - \epsilon))\\
 = \prod_{j\in J}\bpp{f_{Z_j}(\what z_j)\what z_j\cdot2\epsilon + o(\epsilon)}\prod_{k\in J^c}\bpp{F_{Z_j}(\wht z_j)+o(\epsilon)}\,.\label{eq:PZ}
\end{multline}
Observe that for $J\in\calJ_r(A,\vv x)$, the latter expression equals $2w_J\, \epsilon^{|J|} (1 + o(1)),\ \epsilon \downarrow 0$ and the terms with $|J|>r$ will become negligible  since they 
are of smaller order. Therefore, Relation \eqref{eq:PZ} yields \eqref{eq:nu}, and the proof is thus complete. \end{proof}
The proof of Proposition~\ref{prop:limit} provides an insight to the expressions of the weights $w_J$'s in~\eqref{eq:pJ} 
and the components $\nu_J$'s in~\eqref{eq:nuJ}. 
In particular, it explains why only hitting scenarios of rank $r$ are involved in the expression of the conditional probability.
The formal proof of Theorem~\ref{thm:rcp}, however, requires a 
different argument.

\subsection{Conditional Sampling: Computational Efficiency}\label{sec:algorithm}

We discuss here important computational issues related to sampling from the regular conditional 
probability in \eqref{eq:nu}. It turns out that identifying all hitting scenarios amounts to 
solving the {\em set covering problem}, which is NP-hard (see e.g.~\cite{caprara00algorithms}). The probabilistic
structure of the max-linear models, however, will lead us to an alternative efficient solution, valid with probability one.
In particular, we will provide a new formula for the regular conditional probability, showing that $\vv Z$ 
can be decomposed into conditionally independent vectors, given $\vv X = \vv x$. As a consequence, with probability one we are not in the `bad' situation that the corresponding set covering problem requires exponential time to solve. Indeed, this will lead us to 
an efficient and linearly-scalable algorithm for conditional sampling, which works well for max-linear models with large
dimensions $n\times p$ arising in applications.

To fix ideas, observe that Theorem \ref{thm:rcp} implies the following simple algorithm.

\medskip
\noindent{\sc Algorithm I:} 
\begin{enumerate}
\item\label{step:hatz} Compute $\widehat z_j$ for $j = 1,\dots, p$.
\item\label{step:decompose} Identify $\calJ(A,\vv x)$, compute $r = r(\calJ(A,\vv x))$ and focus on the set of {\em relevant hitting scenarios} $\calJ_r = \calJ_r(A,\vv x)$.
\item\label{step:weights} Compute $\{w_J\}_{J\in\calJ_r}$ and $\{p_J\}_{J\in\calJ_r}$.
\item\label{step:resample} Sample $\vv Z\sim\nu(\vv x,\cdot)$ according to~\eqref{eq:nu}.
\end{enumerate}
Step~\ref{step:hatz} is immediate. Provided that Step~\ref{step:decompose} is done, Step~\ref{step:weights} is trivial and, Step~\ref{step:resample} can be carried out by first picking a {\em hitting scenario} $J\in {\cal J}_r(A,\vv x)$ (with probability $p_J(A,\vv x)$), 
setting $Z_j = \widehat z_j$,  for $j\in J$ and then  resampling independently the remaining $Z_j$'s from the truncated distributions: 
$Z_j\mid \{ Z_j < \widehat z_j\}$, for all
 $j\in \{1,\dots,p\}\setminus J$.

The most computationally intensive aspect of
this algorithm is to identify the set of all relevant hitting scenarios $\calJ_r(A,\vv x)$ in Step \ref{step:decompose}.  This is closely related to the NP-hard {\em set covering problem} in theoretical computer science (see e.g.~\cite{caprara00algorithms}), which is formulated next. Let $H = (h_{i,j})_{n\times p}$ be a matrix of $0$'s and $1$'s, and let
$c = (c_j)_{j=1}^p\in\mathbb Z_+^p$ be a $p$-dimensional {\em cost} vector. For simplicity, introduce the notation:
$$
\lb m\rb  \equiv \{1,2,\dots,m\},\ \ m\in \mathbb N.
$$
For the matrix $H$, we say that the column $j\in \lb p\rb $ {\em covers} the row $i\in\lb n\rb $, if $h_{i,j} = 1$. The goal of the set-covering problem is to find a minimum-cost 
subset $J\subset\lb p\rb $, such that every row is covered by at least one column $j\in J$. This is equivalent
to solving
\equh\label{eq:SCP}
\min_{\substack{\delta_j\in\{0,1\}\\j\in\lb p\rb }}\sum_{j\in\lb p\rb } c_j\delta_j\,,\ \mbox{ subject to }\ 
 \sum_{j\in\lb p\rb }h_{i,j}\delta_j\geq 1\,, i\in\lb n\rb \,.
\eque
We can relate the problem of identifying $\calJ_r(A,\vv x)$ to the set covering problem by defining
\equh\label{eq:hij}
h_{i,j} = \ind_{\{a_{i,j}\what z_j = x_i\}},
\eque
where $A = (a_{i,j})_{n\times p}$ and $\vv x = (x_i)_{i=1}^n$ are as in~\eqref{rep:maxLinear}, and  $c_j = 1\,, j\in\lb p\rb $. It is easy to see that, 
every $J\in\calJ_r(A,\vv x)$ corresponds to a solution of~\eqref{eq:SCP}, and vice versa. Namely, for $\{\delta_j\}_{j\in\lb p\rb }$ 
minimizing~\eqref{eq:SCP}, we have $J = \{j\in\lb p\rb :\delta_j = 1\}\in\calJ_r(A,\vv x)$.

The set $\calJ_r(A,\vv x)$ corresponds to the set of {\em all} solutions of~\eqref{eq:SCP}, which depends only 
on the matrix $H$. Therefore, in the sequel we write $\calJ_r(H)$ for $\calJ_r(A,\vv x)$, and 
\begin{equation}\label{eq:H}
H = (h_{i,j})_{n\times p}\equiv {\mathbb H}(A,\vv x),
\end{equation}
with $h_{i,j}$ as in \eqref{eq:hij} will be referred to as the {\em hitting matrix}.

\begin{Example}\label{Ex:1continued}
Recall Example \ref{Ex:1}. The following hitting matrices correspond to the three cases of $\vv x$ 
discussed therein:
\[
H\topp i = \left(
\begin{array}{ccc}
1 & 0 & 0 \\
0 & 1 & 0 \\
0 & 0 & 1 
\end{array}
\right),\ 
H\topp {ii} = \left(
\begin{array}{ccc}
1 & 0 & 0 \\
1 & 1 & 0 \\
0 & 0 & 1
\end{array}
\right)\ \mand \
H\topp {iii} = \left(
\begin{array}{ccc}
1 & 0 & 0 \\
1 & 1 & 0 \\
1 & 1 & 1
\end{array}
\right)
\,.
\]
\end{Example}

Observe that solving for $\calJ_r(H)$ is even more challenging than
solving the set covering problem \eqref{eq:SCP}, where only one minimum-cost subset $J$ is
needed, and often an approximation of the optimal solution is
acceptable. Here, we need to identify exhaustively all $J$'s such
that~\eqref{eq:SCP} holds. Fortunately, this problem can be substantially simplified, thanks to the probabilistic structure of the max-linear model.

{We first study the distribution of $H$.} In view of \eqref{eq:H}, we have that 
$H = {\mathbb H}(A,\vv X)$, with $\vv X= A\odot \vv Z$, is a random matrix. It will turn out that, with
probability  one, $H$ has a {\em nice} structure, leading to an efficient conditional
sampling algorithm.

For any hitting matrix $H$, we will decompose the set $\lb p\rb \equiv \{1,\dots,p\}$ into 
a certain disjoint union $\lb p\rb  = \bigcup_{s=1}^r\wb J\topp s$. The vectors $(Z_j)_{j\in\wb J\topp {s}},\ 
s = 1,\dots, r$ will turn out to be conditionally independent (in $s$), given $\vv X = \vv x$.  Therefore,
$\nu(\vv x,E)$ will be expressed as a product of (conditional) probabilities. 

We start by decomposing the set $\lb n\rb \equiv \{1,\dots,n\}$. First, for all 
$i_1,i_2\in\lb n\rb \,, j\in\lb p\rb $, we write $i_1\stackrel{j}\sim i_2\,,\mbox{ if } h_{i_1,j} = h_{i_2,j} = 1$.
Then, we define an equivalence relation on $\lb n\rb$:
\equh\label{eq:equivalence}
i_1\sim i_2,\  \mbox{ if } \
i_1 = \wt i_0\stackrel{j_1}\sim \wt i_1\stackrel{j_2}\sim\cdots\stackrel{j_m}\sim\wt i_m = i_2\,,
\eque
with some $m\leq n, i_1 = \wt i_0,\wt i_1,\dots,\wt i_m = i_2\in\lb n\rb , j_1,\dots,j_m\in\lb p\rb $.
That is, `$\sim$' is the transitive closure of `$\stackrel j\sim$'. Consequently, we obtain a partition of $\lb n\rb $, denoted by 
\equh\label{eq:[n]}
\lb n\rb  = \bigcup_{s=1}^rI_s\,,
\eque
where $I_s, s = 1,\dots,r$ are the equivalence classes w.r.t.~\eqref{eq:equivalence}.
Based on~\eqref{eq:[n]}, we define further
\eqnhn
J\topp s & = & \bccbb{j\in\lb p\rb : h_{i,j} = 1\mfa i\in I_s}\label{eq:J^s}\,,\\
\wb J\topp s & = & \bccbb{j\in\lb p\rb : h_{i,j} = 1\mbox{ for some } i\in I_s}\,.\label{eq:wbJ^s}
\eqnen
The sets $\{J\topp s, \wb J\topp s\}_{s\in \lb r\rb }$ will determine the factorization form of $\nu(\vv x,E)$. 

\begin{Thm}\label{thm:factorization}
Let $\vv Z$ be as in Theorem~\ref{thm:rcp}. Let also $H$ be the hitting matrix corresponding to $(A,\vv X)$ with $\vv X = A\odot \vv Z$, and $\{J\topp s, \wb J\topp s\}_{s\in \lb r\rb }$ be the sets defined in~\eqref{eq:J^s} and~\eqref{eq:wbJ^s}. Then, with probability one, we have
\begin{itemize}
\item [(i)] $r = r(\calJ(A,\vv X))$,
\item [(ii)] for all $J\subset\lb p\rb $, $J\in \calJ_r(A,A\odot\vv Z)$ if and only if $J$ can be written as 
\equh\label{eq:Jr}
J = \{j_1,\dots,j_r\}\qmwith j_s\in J\topp s\,, s\in\lb r\rb \,,
\eque
\item [(iii)] for $\nu(\vv x,E)$ defined in~\eqref{eq:nu},
\equh\label{eq:nus}
\nu(\vv X,E) = \prod_{s=1}^r\nu\topp s(\vv X,E)\mwith \nu\topp s(\vv X,E) = \frac{\sum_{j\in J\topp s}w\topp s_j(\vv X) \nu\topp s_j(\vv X,E)}
  {\sum_{j\in J\topp s}w\topp s_j(\vv X)}\,,
\eque
where for all $j\in J\topp s$,
\eqnhn
w_j\topp s(\vv x) & \defe & \what z_jf_{Z_j}(\what z_j)\prod_{k\in \wb J\topp s\setminus\{j\}} F_{Z_k}(\what z_k)\label{eq:nuws}\,,\\
\nu_j\topp s(\vv x,E) & \defe & \delta_{\pi_j(E)}(\what z_j)\prod_{k\in \wb J\topp s\setminus\{j\}} \proba(Z_k\in \pi_k(E)|Z_k<\what z_k),\label{eq:nujs}
\eqnen
with $\what z_j = \what z_j(\vv x)$ as in \eqref{eq:zhat}.
\end{itemize}
\end{Thm}
The proof of Theorem~\ref{thm:factorization} is given in Section~\ref{sec:proofs}.
\begin{Rem}
Note that this result does not claim that 
$\nu(\vv x, E)$ in \eqref{eq:nus} is the {\em regular conditional probability}. It merely provides an equivalent
expression for \eqref{eq:nu}, which is valid with probability one.  We still need to show that \eqref{eq:nu}, or equivalently~\eqref{eq:nus}, is indeed
the regular conditional probability.
\end{Rem} 
From~\eqref{eq:nuws} and~\eqref{eq:nujs}, one can see that $\nu\topp s$ is the conditional distribution of $(Z_j)_{j\in\wb J\topp s}$. Therefore, Relation~\eqref{eq:nus} implies that
$\{(Z_j)_{j\in\wb J\topp s}\}_{s\in \lb r\rb }$, as vectors indexed by $s$, are {\em conditionally independent},
given $\vv X = \vv x$. This leads to the following improved conditional sampling algorithm:\medskip

\noindent {\bf Algorithm II:}
\begin{enumerate}
\item\label{step:hatz2} Compute $\widehat z_j$ for $j = 1,\dots, p$ and the hitting matrix $H = {\mathbb H}(A,\vv x)$.
\item\label{step:decompose2} Identify $\{J\topp s,\wb J\topp s\}_{s\in\lb r\rb }$ by~\eqref{eq:J^s} and~\eqref{eq:wbJ^s}.
\item\label{step:weights2} Compute $\{w_j\topp s\}_{j\in J\topp s}$ for all $s\in\lb r\rb $ by~\eqref{eq:nuws}.
\item\label{step:resample2} Sample $(Z_j)_{j\in\wb J\topp s}\mid \vv X = \vv x\sim\nu\topp s(\vv x,\cdot)$ independently for $s = 1,\dots,r$. 
\item\label{step:combine2} Combine the sampled $(Z_j)_{j\in\wb J\topp s}, s = 1,\dots, r$ to obtain a sample $\vv Z$.
\end{enumerate}

This algorithm identifies all hitting scenarios in an efficient way. To illustrate its efficiency compared to Algorithm I, consider that $r = 10$ and 
$|J\topp s| = 10$ for all $s\in\lb 10\rb $. Then, applying Formula~\eqref{eq:nu} in Algorithm I requires storing in memory 
the weights of all $10^{10}$ hitting scenarios. In contrast, the implementation of \eqref{eq:nus} requires saving
only $10\times10$ weights. This improvement is critical in practice since it allows us to handle large, realistic
models. 

Table \ref{table:speed} demonstrates the running times of Algorithm II as a function of the dimensions $n\times p$ of the 
matrix $A$. It is based on a discretized 2-d Smith model (Section \ref{sec:Smith}) and measured on
an {\tt Intel(R) Core(TM)2 Duo CPU E4400 \@2.00GHz} with {\tt 2GB RAM}. It is remarkable that the times scale linearly
in both $n$ and $p$.
  
\begin{table}
\begin{center}
\caption{\label{table:speed} Means and standard deviations (in parentheses) of the running times (in seconds) 
for the decomposition of the hitting matrix $H$, based on 100 independent observations 
$\vv X = A \odot \vv Z$, where $A$ is an $(n\times p)$ matrix corresponding to a discretized Smith model.}
\begin{tabular}{c|cccc}
\hline
 $p\setminus n$  & 1 & 5 & 10 & 50\\
\hline
2500 & 0.03 (0.02) & 0.13 (0.03) &0.24 (0.04) &1.25 (0.09)\\
10000 &0.11 (0.04) & 0.50 (0.05)  & 1.00 (0.08)& 4.98 (0.33)\\
\hline
\end{tabular}
\end{center}
\end{table}

\section{Examples} 

\subsection{MARMA processes}\label{sec:MARMA}
In this section, we apply our result to the max-autoregressive moving average (MARMA) processes studied by Davis and Resnick~\cite{davis89basic}.
A stationary process $\{X_t\}_{t\in\mathbb Z}$ is a MARMA$(m,q)$ process if it satisfies the MARMA recursion:
\equh\label{eq:MARMA}
X_t = \phi_1 X_{t-1}\vee\cdots\vee\phi_mX_{t-m}\vee Z_t\vee\theta_1Z_{t-1}\vee\cdots\vee\theta_qZ_{t-q}\,,
\eque
for all $t\in\mathbb Z$, where $\phi_i\geq 0,\theta_j\geq 0, i= 1,\dots,m, j=1,\dots,q$ are the parameters, and $\{Z_t\}_{t\in\mathbb Z}$ are i.i.d.~1-Fr\'echet random variables. 
Proposition 2.2 in~\cite{davis89basic} shows that,~\eqref{eq:MARMA} has a unique solution in form of
\equh\label{eq:causal}
X_t = \bigvee_{j=0}^\infty\psi_jZ_{t-j}<\infty\,,\mbox{almost surely,}
\eque
with $\psi_j\geq 0, j\geq 0, \sum_{j=0}^\infty \psi_j<\infty$, if and only if $\phi^* = \bigvee_{i=1}^m\phi_i<1$.
In this case, 
\[
\psi_j = \bigvee_{k=0}^{j\wedge q}\alpha_{j-k}\theta_k\,,
\]
where $\{\alpha_j\}_{j\in\mathbb Z}$ are determined recursively by $\alpha_j = 0$ for all $j<0$, $\alpha_0 = 1$ and 
\equh\label{eq:alphaj}
\alpha_j = \phi_1\alpha_{j-1}\vee\phi_2\alpha_{j-2}\vee\cdots\vee\phi_m\alpha_{j-m}\,,\forall j\geq 1\,.
\eque
In the sequel, we will focus on the MARMA process~\eqref{eq:MARMA} with unique stationary solution~\eqref{eq:causal}. In this case, the MARMA process is a spectrally discrete max--stable process. Without loss of generality, we also assume $\indz Z$ to be standard 1-Fr\'echet.

We consider the prediction of the MARMA process in the following framework: suppose at each time $t\in\{1,\dots,n\}$ we observe the value $X_t$ of the process, and the goal is to predict $\{X_s\}_{n<s\leq n+N}$. We do so by generating i.i.d.~samples from the conditional distribution $\{X_s\}_{n<s\leq n+N}\mid\{X_t\}_{t = 1,\dots,n}$. 
To apply our result, it suffices to provide a max-linear representation of this model. We will truncate~\eqref{eq:causal} to obtain
\equh\label{eq:wtX}
\wt X_t = \bigvee_{j=0}^p\psi_jZ_{t-j}\,,\forall t = 1,\dots,n+N\,.
\eque
The truncated process can approximate the original one arbitrarily well, if we take $p$ large enough. Indeed, by using the independence and max-stability of the $Z_t$'s, one can show that
\equh\label{eq:wtXt=Xt}
\proba(\wt X_t = X_t) 
= \proba\bpp{\bigvee_{j=0}^p\psi_jZ_{t-j}\geq \bigvee_{j=p+1}^\infty\psi_jZ_{t-j}} = 1-\frac{\sum_{j=p+1}^\infty\psi_j}{\sif j0\psi_j}\longrightarrow 1\,,
\eque
as $p\to\infty$. Moreover, by induction on $\alpha_j$ in~\eqref{eq:alphaj}, one can show that $\alpha_j\leq(\phi^*)^{\lceil j/m\rceil}$ for all $j\in\mathbb N$, and thus the convergence~\eqref{eq:wtXt=Xt} above is geometrically fast.

Now, we reformulate the prediction problem with the model~\eqref{eq:wtX} as follows:
\eqnh
\mbox{\em observe } {\vv X}_{[1,n]} = A\odot \vv Z,\qmand \mbox{\em predict }  \vv Y_{[1,N]} = B\odot \vv Z\mid\vv X_{[1,n]}\,,
\eqne
with the notation $\vv X_{[1,n]} = (\wt X_1,\dots,\wt X_n)$, $\vv Y_{[1,N]} = (\wt X_{n+1},\dots,\wt X_{n+N})$ and $\vv Z= (Z_{1-p},Z_{2-p},\dots,Z_{n+N})$. Here, $A\in\mathbb R_+^{n\times(p+n+N)}, B\in\mathbb R_+^{N\times(p+n+N)}$ are determined by~\eqref{eq:wtX}.  
In particular, 
\equh\label{eq:AB}
\left(
\begin{array}{c}
A\\
B
\end{array}
\right) = \left(
\begin{array}{cccccccc}
\psi_p&\psi_{p-1}& \cdots     & \psi_0 & 0        &0         &\cdots & 0\\
0  	  &\psi_p    & \psi_{p-1} & \cdots &\psi_0    &0         &\cdots &0\\
\vdots&    \ddots      & \ddots     & \ddots &          &\ddots    &\ddots &\vdots\\
0     &\cdots    & 0          & \psi_p &\psi_{p-1}&\cdots    &\psi_0 & 0\\
0     &\cdots    & 0          & 0      & \psi_p   &\psi_{p-1}&\cdots &\psi_0 
\end{array}
\right)\,.
\eque
In practice, given the observations $\vv X_{[1,n]}$, we use our algorithm to sample from the conditional distribution $\vv Z\mid\vv X_{[1,n]}$. Therefore, we can sample
\equh\label{eq:Y}
\vv Y_{[1,N]}\mid\vv X_{[1,n]} \eqd \vv B\odot\vv Z\mid \vv X_{[1,n]}\,.
\eque

Our approach is different from the prediction considered in~\cite{davis89basic}, which we will briefly review. Davis and Resnick took the classic time series point of view and investigated how to approximate $X_s$ by a max-linear combination of $\{X_t\}_{t = 1,\dots,n}$, w.r.t.~a certain metric $d$. Namely, for all $Y\in\calH$ with
\[
\calH = \bccbb{\bigvee_{j=-\infty}^\infty\alpha_jZ_j:\alpha_j\geq 0,\sum_{j=-\infty}^\infty\alpha_j<\infty}\,,
\]
they considered a projection of $Y$ onto the space $\filF_n$, max-linearly spanned by $\{X_t\}_{t=1,\dots,n}$: $\filF_n = \sccbb{\bigvee_{j=0}^\infty b_jX_{n-j}:b_j\geq 0, \sum_{j=0}^\infty b_j<\infty}$.
That is, consider the projection $\calP_nY$ defined by
\equh\label{eq:PnY}
\calP_nY = \argmin_{\wt Y\in\filF_n}d(\wt Y,Y)
\eque
with the metric $d$ induced by $d\spp{\bigvee_j\alpha_jZ_j,\bigvee_j\beta_jZ_j} = \sum_j|\alpha_j-\beta_j|$.
For specific MARMA processes,~\cite{davis89basic} provided predictors based on the projection~\eqref{eq:PnY}. We will refer to these predictors as the {\it projection predictors}. 

In general, the conditional samplings reflect the conditional distribution~\eqref{eq:Y}, and they provide more information than the
projection predictors. Sampling multiple times from~\eqref{eq:Y}, we can calculate e.g., conditional medians, conditional means, quantiles, etc., which 
are optimal predictors with respect to various loss functions.

\begin{Example}[MAR$(m)$ processes]
Consider the MAR$(m)\equiv$MARMA$(m,0)$ process with 
\equh\label{eq:directCS}
X_t = \phi_1X_{t-1}\vee\cdots\vee\phi_mX_{t-m}\vee Z_t\,.
\eque
The projection predictor for this model can be obtained recursively by
\equh\label{eq:DR}
\what X_{t+k} = \phi_1\what X_{t+k-1}\vee\cdots\vee \phi_m\what X_{t+k-m}\,,
\eque
with $\what X_t = X_t, t = 1,\dots, n$ (see \cite{davis89basic}, p.\ 799).
\end{Example}
Figure~\ref{fig:MARMA} illustrates an application of our conditional sampling algorithm in this case. Consider an MAR(3) process $\{X_t\}_{t = 1}^{150}$ with $\phi_1 = 0.7, \phi_2 = 0.5$ and $\phi_3 = 0.3$. In effect, we use the truncated model $\{\wt X_t\}_{t\in\mathbb N}$ in~\eqref{eq:wtX} with $p = 500$, but we still write $X_t$ for the sake of simplicity. Treating the first 100 values as observed,  we plot the projection predictor, conditional upper $95\%$-quantiles and the conditional medians of $\{X_s\}_{s = 101}^{150}$ based on 500 independent samples from the conditional distribution.

Observe that the value of the projection predictor in Figure~\ref{fig:MARMA} is always below the conditional median. This ``underestimation'' phenomenon was typical in all the simulations we performed. It can be explained by the fact that, the projection predictor in~\eqref{eq:DR} does not account for the jumps of the process caused by new arrivals $\{Z_t\}_{t>100}$. Indeed, a large new arrival $Z_t$ will cause the process to jump immediately to $Z_t$ at time $t$, but this will never occur for the projection predictor $\what X_t$. 
 \begin{figure} [!ht]
  \begin{center}
  \ifthenelse{\boolean{showFigure}}{
   \includegraphics[width = .9\textwidth]{MARMA.pdf}
   }{\includegraphics[width = .9\textwidth]{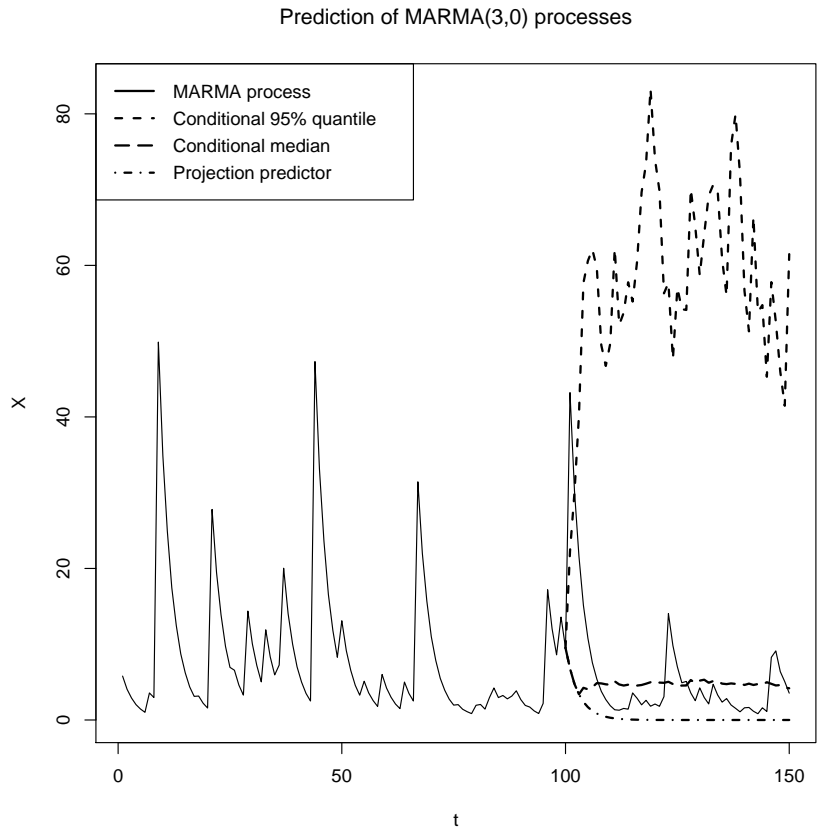}}
   {\caption{\label{fig:MARMA} Prediction of a MARMA(3,0) process with $\phi_1 = 0.7, \phi_2 = 0.5$ and $\phi_3 = 0.3$, based on the observation of
the first 100 values of the process.}}
  \end{center}
 \end{figure}

Next, we apply our algorithm to examine the bias of the projection predictor. To do this, for each generated MARMA process, we calculated the cumulative probability that the projection predictor corresponds to, for each location $s = 101,\dots,150$.
Namely, using 500 independent samples $\{X_s\topp k\}_{s = 101}^{150}, k = 1,\dots,500$ from the conditional distribution, we calculated
\equh\label{eq:QP}
\proba(X_s\leq \what X_s\mid\{X_t\}_{t=1}^{100}) \approx \frac1{500}\sum_{k=1}^{500} \ind{\{X_s\topp k\leq \what X_s\}}, \forall s>100\,,
\eque
where $\what X_s$ is the projection predictor in~\eqref{eq:DR}. This procedure was repeated 1000 times for independent realizations of $\{X_t\}_{t=1}^{100}$ and the means of the (estimated) probability in~\eqref{eq:QP} are reported in Table~\ref{tab:AQP}. Note that as the time lag increases, the conditional quantiles of the projection predictors decrease. In this way, our conditional sampling algorithm helps quantify numerically the observed underestimation phenomenon in Figure~\ref{fig:MARMA}.
\begin{table}[!ht]
\begin{center}
\caption{\label{tab:AQP} Cumulative probabilities that the projection predictors correspond to at time $100 + t$, based on 1000 simulations.}
\begin{tabular}{c|ccccccccccc}
 $t$  & 1 & 2 & 3 & 4 & 5 & 10 & 20 & 30 & 40 \\
\hline
mean & $70.6\%$ & $50.3\%$ & $35.6\%$ & $25.3\%$  & $17.8\%$ & $2.9\%$ & $0.1\%$ & $0\%$ & $0\%$ 
\end{tabular}
\end{center}
\end{table}


Finally, we compare the generated conditional samples to the true process values at times $s = 101,\dots,150$. Our goal is to demonstrate the validity of our conditional sampling algorithm. The idea is that, at each location $s = 101,\dots,150$, the true process should lie below the predicted $95\%$ upper
confidence bound of $X_s\mid \{X_t\}_{t=1}^{100}$, with probability at least $95\%$. (Note that due to the presence of atoms in the conditional distributions,
the coverage probability may in principle be higher than $95\%$.) Motivated by this, we repeat the procedure in the previous paragraph and record the proportion of the times that $X_s$ is below the predicted confidence quantile, for each $s$. We refer to these values as the {\it coverage rates}. As discussed, the coverage rates should be close to $95\%$. This is supported by our simulation result, shown in Table~\ref{tab:CR}. 
\begin{table}[!ht]
\begin{center}
\caption{\label{tab:CR} Coverage rates (CR) and the widths of the upper $95\%$ confidence intervals at time $100 + t$, based on 1000 simulations.
}

\begin{tabular}{c|ccccccccc}
 $t$  & 1 & 2 & 3 & 4 & 5 & 10 & 20 & 30 & 40\\
\hline
CR & 0.956 & 0.952 & 0.954 & 0.957 & 0.966 & 0.947 & 0.943 & 0.951 & 0.955\\
 width & 13.06 & 26.6 & 37.8 & 45.6 & 51.2 & 62.8 & 66.0 & 66.2 & 65.4
\end{tabular}
\end{center}
\end{table}

Table~\ref{tab:CR} also shows the widths of the upper $95\%$-confidence intervals. Note that these widths are not equal to the upper confidence 
bounds, given by the conditional $95\%$-quantiles, since the left end-point of the conditional distributions are greater than zero.
When the time lag is small, the left end-point is large and the widths are small, due to the strong influence of the past of the process
$\{X_t\}_{t=1}^{100}$. On the other hand, because of the weak temporal dependence of the MAR(3)
processes, this influence decreases fast as the lags increase. Consequently, the conditional distribution converges to the unconditional one, and the
conditional quantile to the unconditional one. Note that the (unconditional) $95\%$-quantile of $X_s$ for the MARMA process~\eqref{eq:causal} can
be calculated via the formula $0.95 = \proba(\sigma Z\leq u) = \exp(-\sigma u\inv)$, with $\sigma = \sum_{j=0}^{p}\psi_j$. For the MAR(3) process we
chose, we have $\sigma = 3.4$ and the $95\%$-quantile of $X_s$ equals $66.29$. This is consistent with the widths in Table \ref{tab:CR} for large lags.
\begin{Rem}
As pointed out by an anonymous referee, in this case one can directly generate samples from $\{X_s\}_{s=n+1}^N\mid\{X_t\}_{t=1}^n$, by generating independent Fr\'echet random variables and iterating~\eqref{eq:directCS}. We selected this example only for illustrative purpose and to be able to compare with the projection predictors in~\cite{davis89basic}. One can modify slightly the prediction problem, such that our algorithm still applies by adjusting accordingly~\eqref{eq:AB}, while both the projection predictor and the direct method by using~\eqref{eq:directCS} do not apply. For example, consider the prediction problem with respect to the conditional distribution $\proba(\{X_s\}_{s=2n+1}^{2n+N}\in\cdot\mid \{X_t:t = 1,3,\dots,2n-1\})$ (prediction with only partial history observed) or $\proba(\{X_s\}_{s=2}^{n-1}\in\cdot\mid X_1,X_n)$ (prediction of the middle path with the beginning and the end-point (in the future) given). In other words, our algorithm has no restriction on the locations of observations. This feature is of great importance in spatial prediction problems. 
\end{Rem}

\subsection{The Discrete Smith Model}\label{sec:Smith}

Consider the following {\em moving maxima} random field model in $\mathbb R^2$:
 \begin{equation}\label{e:2d-X}
   X_{\vv t} =  \Eint{\mathbb R^2} \phi(\vv t -\vv u) M_\alpha(\d \vv u ), \ \ \ \vv t = (t_1,t_2) \in \mathbb R^2,
   \end{equation}
 where $M_\alpha$ is an $\alpha$-Fr\'echet random sup-measure on $\mathbb R^2$ with the Lebesgue control 
 measure. Smith \cite{smith90max} proposed to use for $\phi$ the bivariate Gaussian density:
  \begin{equation}\label{e:phi}
  \phi(t_1,t_2) \defe \frac{\beta_1\beta_2}{2\pi\sqrt{1-\rho^2}}
  \exp\bccbb{ - \frac1{2(1-\rho^2)}\bb{\beta_1^2t_1^2 - 2\rho\beta_1\beta_2t_1t_2 + \beta_2^2t_2^2}},
 \end{equation}
with correlation $\rho\in (-1,1)$ and variances $\sigma_i^2 = 1/\beta_i^2,\ i=1,2$. Consistent and asymptotically 
normal estimators for the parameters $\rho,\ \beta_1$ and $\beta_2$ were obtained by de Haan and Pereira~\cite{dehaan06spatial}.
Here, we will assume that these parameters are known and will illustrate the conditional sampling methodology 
over a discretized version of the random field \eqref{e:2d-X}. 
Namely, we truncate the extremal integral in \eqref{e:2d-X} to the square region $[-M,M]^2$ and consider a uniform 
mesh of size $h\defe M/q,\ q\in \mathbb N$. We then set
\begin{equation}\label{e:2d-X-discrete}
 X_{\vv t} \defe \bigvee_{-q\le j_1,j_2\le q-1}h^{2/\alpha} \phi(\vv t - \vv u_{j_1 j_2}) Z_{j_1j_2},
\end{equation}
where $\vv u_{j_1j_2} = ((j_1+1/2)h,(j_2+1/2)h)$ and $h^{2/\alpha} Z_{j_1j_2} \stackrel{d}{=}
M_{\alpha}( (j_1h,(j_1+1)h]\times(j_2h,(j_2+1)h])$. This discretized model~\eqref{e:2d-X-discrete} can be made arbitrarily close to 
the spectrally continuous one in \eqref{e:2d-X} by taking a fine mesh $h$ and sufficiently large $M$ (see e.g.~\cite{stoev06extremal}).

Suppose that the random field $X$ in \eqref{e:2d-X-discrete} is observed at  $n$ locations 
$X_{\vv t_i} = x_i,\ \vv t_i \in [-M,M]^2,\ i = 1,\dots, n$. In view of \eqref{e:2d-X-discrete}, we have 
the max-linear model $\vv X = A \odot \vv Z$, with $\vv X = (X_{\vv t_i})_{i=1}^n$ and 
$\vv Z = (Z_{j})_{j=1}^p,\ p=q^2$. By sampling from the conditional distribution of $\vv Z \mid \vv X = \vv x$, we can predict the random field $X_{\vv s}$ at arbitrary locations $\vv s
\in{\mathbb R}^2$. 

To illustrate our algorithm, we used the model~\eqref{e:2d-X-discrete} with
parameter values 
$\rho = 0, \beta_1 = \beta_2 = 1, M = 4, p = q^2 = 2500$, and $n = 7$ observed locations. 
We generated $N=500$ independent samples from the conditional distribution of the random field $\{X_{\vv s}\}$, 
where $\vv s$ takes values on an uniform $100\times100$ grid, in the region $[-2,2]\times[-2,2]$. We have already seen four of these realizations in
Figure~\ref{fig:samplings}. 
Figure~\ref{fig:med_quant} illustrates the median and $0.95$-th quantile of the conditional distribution. The
former provides the optimal predictor for the values of the random field given the observed data,
with respect to the absolute deviation loss. The marginal quantiles, on the other hand, provide important confidence
regions for the random field, given the data. 

Certainly, conditional sampling may be used to address more complex
{\em functional} prediction problems.  In particular, given a two-dimensional threshold surface, one can readily obtain 
the {\em correct} probability that the random field exceeds or stays below this surface, conditionally on the
observed values. This is much more than what marginal conditional distributions can provide. 
 \begin{figure} [!ht]
  \begin{center}
  \ifthenelse{\boolean{showFigure}}{
   \includegraphics[width = .45\textwidth]{CS_Med.pdf}\includegraphics[width = .45\textwidth]{CS_MQ.pdf}
      }{\includegraphics[width = .45\textwidth]{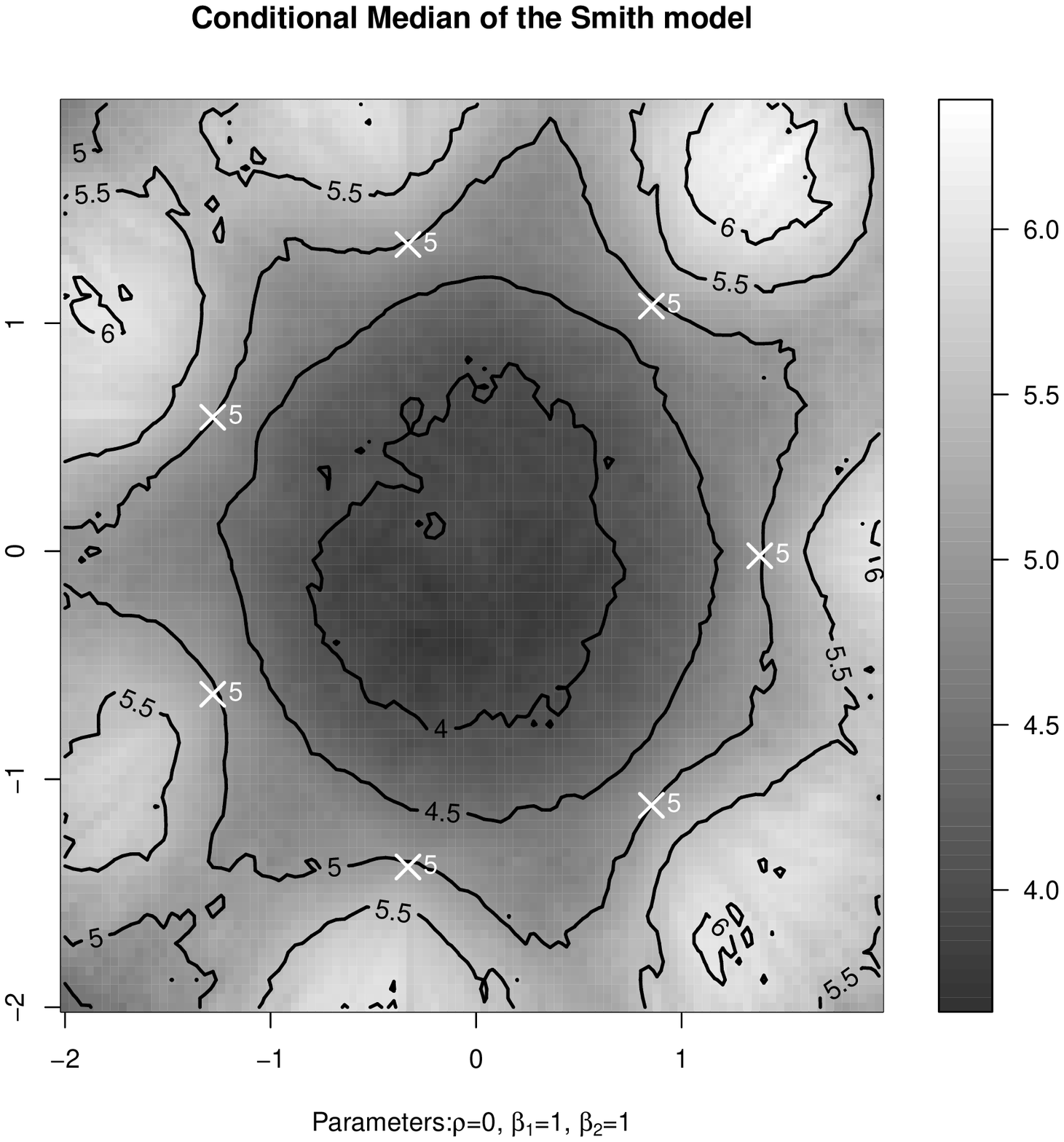}\includegraphics[width= .45\textwidth]{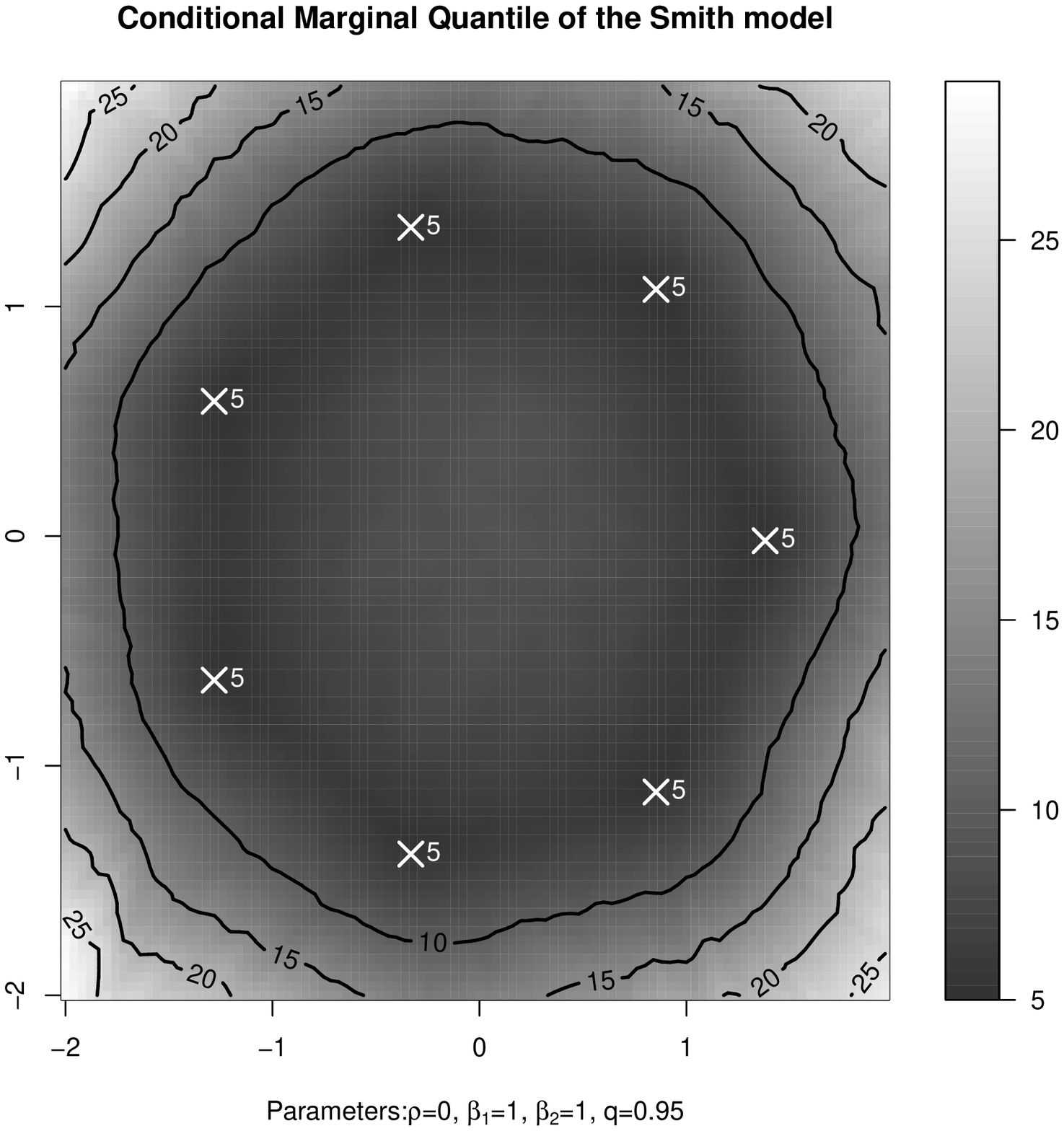}}
   {\caption{\label{fig:med_quant} Conditional medians ({\em left}) and $0.95$-th conditional marginal quantiles 
   ({\em right}). Each cross indicates an {\em observed} location of the random field, with the observed value at right.}}
  \end{center}
 \end{figure}


\section{Proofs of Theorems~\ref{thm:rcp} and~\ref{thm:factorization}} \label{sec:proofs}
In this section, we prove Theorems~\ref{thm:rcp} and~\ref{thm:factorization}. 
We will first prove Theorem~\ref{thm:factorization}, which simplifies the regular conditional probability formula~\eqref{eq:nu} in Theorem~\ref{thm:rcp}. Then, we show the simplified new formula is the desired regular conditional probability, which completes the proof of Theorem \ref{thm:rcp}. 
The key step to prove Theorem~\ref{thm:factorization} is the following lemma.
Write $H_{\cdot j} = \{i\in\lb r\rb: h_{i,j} = 1\}$.
\begin{Lem}\label{lem:H}
Under the assumptions of Theorem~\ref{thm:factorization}, with probability one, 
\begin{itemize}
\item[(i)] $J\topp s$ is nonempty for all $s\in\lb r\rb $, and
\item[(ii)] for all $j\in J\topp s$, $H_{\cdot j}\cap I_s \neq \emptyset$ implies $H_{\cdot j}\subset I_s$.
\end{itemize}
\end{Lem}
\begin{proof}
Note that to show part (ii) of Lemma~\ref{lem:H}, it suffices to observe that since $I_s$ is an equivalence class w.r.t.\ Relation~\eqref{eq:equivalence}, $H_{\cdot j}\setminus I_s$ and $H_{\cdot j}\cap I_s$ cannot be both nonempty. Thus, it remains to show part (i). We proceed by excluding several $\proba$-measure zero sets, on which the desired results may not hold.

First, observe that for all $i\in\lb n\rb $, the maximum value of $\{a_{i,j}Z_j\}_{j\in\lb r\rb }$ is achieved for unique $j\in\lb p\rb $ with probability one, since the $Z_j$'s are independent and have continuous distributions. Thus, the set 
\[
\calN_1 \defe\bigcup_{{i\in\lb n\rb ,j_1,j_2\in\lb p\rb, j_1\neq j_2 }}\bccbb{a_{i,j_1}Z_{j_1} = a_{i,j_2}Z_{j_2} = \max_{j\in\lb p\rb }a_{i,j}Z_j}
\]
has $\proba$-measure zero. From now on, we focus on the event $\calN_1^c$ and set $j(i) = \argmax_{j\in\lb p\rb }a_{i,j}Z_j$ for all $i\in\lb n\rb $. 

Next, we show that with probability one, $i_1\stackrel j\sim i_2$ implies $j(i_1) = j(i_2)$. That is, the set
\[
\calN_2 \defe\bigcup_{{j\in \lb p\rb ,i_1,i_2\in\lb n\rb, i_1\neq i_2 }}\calN_{j,i_1,i_2}\mwith \calN_{j,i_1,i_2}\defe \bccbb{j(i_1)\neq j(i_2), i_1\stackrel j\sim i_2}
\]
has $\proba$-measure 0. It suffices to show $\proba(\calN_{j,i_1,i_2}) = 0$ for all $i_1\neq i_2$. If not, since $\lb p\rb $ and $\lb n\rb $ are finite sets, there exists $\calN_0\subset \calN_{j,i_1,j_2}$, such that $j(i_1) = j_1\neq j(i_2) = j_2$ on $\calN_0$, and $\proba(\calN_0)>0$. At the same time, however, observe that $i_1\stackrel j\sim i_2$ implies $h_{i_1,j} = h_{i_2,j} = 1$, which yields
\[
a_{i_k,j}\what z_j = x_{i_k} = a_{i_k,j(i_k)}Z_{j(i_k)} = a_{i_k,j_k}Z_{j_k}\,, k = 1,2\,.
\]
It then follows that on $\calN_0$, $Z_{j_1}/Z_{j_2} = a_{i_1,j}a_{i_2,j_2}/(a_{i_2,j}a_{i_1,j_1})$, which is a constant. This constant is strictly positive and finite. Indeed, this is because on $\calN_1^c$, $a_{i,j(i)}>0$ by Assumption A and $h_{i,j} = 1$ implies $a_{i,j}>0$. Since $Z_{j_1}$ and $Z_{j_2}$ are independent continuous random variables, it then follows that $\proba(\calN_0) = 0$.

Finally, we focus on the event $(\calN_1\cup \calN_2)^c$. Then, for any $i_1,i_2\in I_s$, we have $i_1\sim i_2$ and let $\wt i_0,\dots,\wt i_n$ be as in~\eqref{eq:equivalence}. It then follows that $j(i_1) = j(\wt i_0) = j(\wt i_1) = \cdots =j(\wt i_n) = j(i_2)$. Note that for all $i\in \lb n\rb $, $h_{i,j(i)} = 1$ by the definition of $j(i)$. Hence, $ j(i_1) = j(i_2)\in J\topp s$. We have thus completed the proof.
\end{proof}
\begin{proof}[Proof of Theorem~\ref{thm:factorization}]
Since $\{I_s\}_{s\in \lb r\rb }$ are disjoint with $\bigcup_{s\in\lb r\rb }I_s =\lb n\rb $, in the language of the set-covering problem, to cover $\lb n\rb $, we need to cover each $I_s$. By part (ii) of Lemma~\ref{lem:H}, any two different $I_{s_1}$ and $I_{s_2}$ cannot be covered by a single set $H_{\cdot j}$. Thus we need at least $r$ sets to cover $\lb n\rb $. On the other hand, with probability one we can select one $j_s$ from each $J\topp s$ (by part (i) of Lemma~\ref{lem:H}), which yields a valid cover. That is, with probability one, $r = r(\calJ(H))$ and any valid minimum-cost cover of $\lb n\rb $ must be as in~\eqref{eq:Jr}, and vice versa. We have thus proved parts (i) and (ii).

To show (iii), by straight-forward calculation, we have, with probability one, 
\eqnhn
\sum_{J\in\calJ_r(A,\vv x)}w_J & = & \sum_{j_1\in J\topp 1}\cdots\sum_{j_r\in J\topp r}w_{j_1,\dots,j_r}\nonumber\\
& = & \sum_{j_1\in J\topp 1}\cdots\sum_{j_{r-1}\in J\topp {r-1}}\Bigg[\prod_{s=1}^{r-1}\what z_{j_s}f_{Z_{j_s}}(\what z_{j_s})\prod_{\substack{j\notin \wb J\topp r\\ j\neq j_1,\dots,j_{r-1}}}F_{Z_j}(\what z_j)\nonumber\\
& & \quad\quad\quad\times \bccbb{\sum_{j\in J\topp r}\bpp{\what z_jf_{Z_{j}}(\what z_{j})\prod_{k\in \wb J\topp r\setminus\{j\}}F_{Z_k}(\what z_k)}}\Bigg]\nonumber\\
& = & \prod_{s=1}^r{\sum_{j\in J\topp s}\bpp{\what z_jf_{Z_j}(\what z_j)\prod_{k\in \wb J\topp s\setminus\{j\}}F_{Z_k}(\what z_k)}} = \prod_{s=1}^r{\sum_{j\in J\topp s}w_j\topp s}\,.\label{eq:plug1}
\eqnen
Similarly, we have
\equh\label{eq:plug2}
\sum_{J\in\calJ_r(A,\vv x)} w_J\nu_J(\vv x,E) = \prod_{s=1}^r\bpp{\sum_{j\in J\topp s}w_j\topp s\nu_j\topp s(\vv x,E)}\,.
\eque
By plugging~\eqref{eq:plug1} and~\eqref{eq:plug2} into~\eqref{eq:nu}, we obtain the desired result and complete the proof.
\end{proof}

%
%
%

\begin{proof}[Proof of Theorem~\ref{thm:rcp}] To prove that $\nu$ in \eqref{eq:nu} yields the regular conditional probability of
$\vv Z$ given $\vv X$, it is enough to show that
\begin{equation}\label{e:P-X,E}
 \proba(\vv X\in D, \vv Z\in E) = \int_{D} \nu(\vv x,E)\proba^{\vv X}(\d \vv x),
\end{equation}
for all rectangles $D\in {\cal R}_{\mathbb R^n_+}$ and $E \in {\cal R}_{\mathbb R_+^p}$. In view of Theorem \ref{thm:factorization},
it is enough to work with $\nu(\vv x, E)$ given by \eqref{eq:nus}.

We shall prove \eqref{e:P-X,E} by breaking the integration into a suitable sum of integrals over regions corresponding to all hitting 
matrices $H$ for the max-linear model $\vv X = A\odot \vv Z$. We say such a hitting matrix $H$ is {\em nice}, if $J\topp s$ defined 
in~\eqref{eq:J^s} is nonempty for all $s\in\lb r\rb $.  In view of Lemma \ref{lem:H}, it suffices to focus on 
the set $\calH(A)$ of {\em nice} hitting matrices $H$. Notice that the set $\calH(A)$ is finite since the elements of the hitting 
matrices are $0$'s and $1$'s.

For all rectangles $D\in\calR_{\mathbb R_+^n}$, let 
\[
D_H = \bccbb{\vv x = A\odot \vv z: \mathbb H(A,\vv x) = H, \vv x\in D}\,
\]
be the set of all $\vv x \in \mathbb R_+^n$ that give rise to the hitting matrix $H$.
By Lemma~\ref{lem:H} {(i)}, for the random vector $\vv X = A \odot \vv Z$, with probability one,  we have
$$
 \vv X = \sum_{H\in \calH(A)} \vv X \ind_{D_H}(\vv X)
$$
and hence 
\equh\label{eq:sumH}
 \int_D\nu(\vv x,E)\proba^{\vv X}(\d \vv x) = \sum_{H\in\calH(A)}\int_{D_H}\nu(\vv x,E)\proba^{\vv X}(\d\vv x)\,.
\eque

Now {fix} an arbitrary and {\em non-random} nice hitting matrix $H\in\calH(A)$. Let $\{I_s\}_{s\in \lb r\rb }$ denote
the partition of $\lb n\rb$ determined by~\eqref{eq:equivalence} and let $J\topp s,\ \wb J\topp s,\ s = 1,\dots, r$ be as in~\eqref{eq:J^s}. 
Recall that $J^{(s)} \subset \wb J^{(s)}$ and the sets $\wb J^{(s)},\ s = 1,\dots, r$ are disjoint.

\medskip
{Focus on the set $D_H \subset \mathbb R_+^n$}. Without loss of generality, and for notational convenience, 
suppose that $s\in I_s$, for all $s = 1,\dots, r$. That is, 
$$
 I_1=\{1,i_{1,2},\dots,i_{1,k_1}\},\ I_2=\{2,i_{2,2},\dots,i_{2,k_2}\},\cdots,\  I_r=\{r,i_{r,2},\dots,i_{r,k_r}\}.
$$

Define the projection mapping $\calP_H: D_H \to \mathbb R_+^r$ onto the first $r$ coordinates:
\[
 \calP_H(x_1,\dots,x_n) = (x_{1},\dots,x_{r}) \equiv \vv x_{\vv r}\,.
\]
Note that $\calP_H$, {\em restricted to} $D_H$ is one-to-one. Indeed, for all $i\in I_s,$ we have $x_i = a_{i,j}\what z_j$ and
$x_{s} = a_{ s,j}\what z_j$, for all $j\in J\topp s$ (recall \eqref{eq:J^s}). This implies $ x_i = (a_{i,j}/a_{s,j}) x_s,$ for all
$i\in I_s$ and all $s = 1,\dots,r$. Hence, $\calP_H(\wtilde {\vv x}) = \calP_H(\vv x)$ implies $\wtilde {\vv x}= \vv x$.

Consequently, can write $\vv x = \calP^{-1}(\vv x_{\vv r}),\ \vv x_{\vv r}\in \calP(D_H)$, and
\[
\int_{D_H}\nu(\vv x,E)\proba^{\vv X}(\d \vv x)
  = \int_{\calP_H(D_H)}\nu( \vv x ,E) \Q_H^{\vv X_{\vv r}}(\d x_{1}\dots\d x_{r}),
\]
where $\Q_H^{\vv X_{\vv r}} \defe \proba^{\vv X}\circ\calP_H\inv$ is the induced measure on the set $\calP_H(D_H)$.

\begin{Lem}\label{lem:density} The measure $\Q_H^{\vv X_{\vv r}}$ has a density with respect to the Lebesgue measure 
on the set $\calP_H(D_H)$. The density is given by
\begin{equation}\label{e:lem:density}
\Q^{\vv X_{\vv r}}_H(\d \vv x_{\vv r}) = 
 \ind_{\calP_H(D_H)}(\vv x_{\vv r})\prod_{s=1}^r \sum_{j\in J^{(s)}} w_j\topp s (\vv x) \frac{ \d x_1}{x_1} \cdots \frac{\d x_r }{x_r}.
\end{equation}
\end{Lem}

\noindent 
The proof of this result is given below. In view of
\eqref{e:lem:density} and \eqref{eq:nus}, we obtain
\begin{eqnarray*}
& & \int_{\calP_H(D_H)}\nu(\vv x,E)\Q^{\vv X_{\vv r}}_H(\d \vv x_{\vv r}) \\
  & & = \int_{\calP_H(D_H)} \underbrace{
 \prod_{s=1}^r \bpp{\frac{ \sum_{j\in J^{(s)}} w_j^{(s)}(\vv x) \nu_j^{(s)} (\vv x, E)}{ \sum_{k\in J^{(s)}} w_k^{(s)}(\vv x) }}}_{ = \nu(\vv x,E)}  
\times \underbrace{\prod_{s=1}^r  \sum_{j\in J^{(s)}} w_j^{(s)}(\vv x) \frac{ \d x_1}{x_1} \cdots \frac{\d x_r }{x_r}}
_{ = \Q_H^{\vv X_{\vv r}} (\d \vv x_{\vv r}) }  \\
& & = \int_{\calP_H(D_H)} \prod_{s=1}^r \sum_{j\in J^{(s)}} w_j^{(s)}(\vv x) \nu_j^{(s)}(\vv x, E) \frac{ \d x_1}{x_1} \cdots \frac{\d x_r }{x_r}\,,
\end{eqnarray*}
which equals
\equh \label{eq:partition}
 \sum_{j_1\in J^{(1)}, \cdots, j_r\in J^{(r)}}\underbrace{\int_{\calP_H(D_H)} \prod_{s=1}^r  w_{j_s}^{(s)}(\vv x) \nu_{j_s}^{(s)}(\vv x, E) \frac{\d x_1}{x_1} \cdots \frac{\d x_r}{x_r}}_{=:I(j_1,\dots,j_r)}\,.
\eque
Fix $j_1\in J^{(1)},\cdots, j_r\in J^{(r)}$ and focus on the integral $I(j_1,\cdots,j_r)$. 
Define
$$
 \Omega_H^r(D_H) \defe
{\Big\{} (z_{j_1},\dots, z_{j_r})\, :\, z_{j_s} = x_s/a_{s,j_s},\ s = 1,\dots, r, \vv x_{\vv r} = (x_s)_{s=1}^r \in \calP_H(D_H) {\Big\}}.
$$
We have, by~\eqref{eq:nuws},~\eqref{eq:nujs}, and replacing $x_s$ with $a_{s,j_s} z_{j_s},\ s = 1,\dots, r$ (simple change of variables),
\begin{eqnarray}
 & & I(j_1,\cdots,j_r) \nonumber\\
 & & \ \ = \int_{\Omega_H^{r}(D_H)} \prod_{s=1}^r {\Big(}
  z_{j_s} f_{Z_{j_s}} (z_{j_s}) \prod_{k\in \wb J^{(s)}\setminus\{j_s\}} F_{Z_k}(\what z_k)\nonumber
\\
 & & \ \ \ \ \ \ \times \delta_{\pi_{j_s}(E)}(z_{j_s}) \prod_{k\in \wb{J}^{(s)}\setminus\{j_s\}} \proba(Z_k \in \pi_k(E) \mid Z_k < \what z_k) 
 {\Big)} \frac{\d z_{j_1}}{z_{j_1}} \cdots \frac{\d z_{j_r}}{z_{j_r}}\nonumber\\
& & \ \ =
 \mathop{\int}_{\Omega_H^{r}(D_H)} \prod_{s=1}^r f_{Z_{j_s}}(z_{j_s}) \delta_{\pi_{j_s}(E)}(z_{j_s}) \nonumber\\
 & & \ \ \ \ \ \ \times \prod_{k\in\lb p \rb \setminus \{j_1,\dots,j_r\}} \proba(Z_k\in \pi_k(E), Z_k <
 \what z_k) \d z_{j_1} \cdots \d z_{j_r}\,.\label{eq:Ij}
\end{eqnarray}
Define 
\begin{multline*}
\Omega_{H; j_1,\dots,j_r}(D_H) = {\Big\{}\vv z\in \mathbb R_+^p\, :\, \vv x = A\odot \vv z\in D_H\,,\\
z_{j_s} = x_s/a_{s,j_s},  s = 1,\dots, r, z_k < \what z_k(\vv x),\ k\in \lb p\rb \setminus\{j_1,\dots,j_r\} {\Big\}}.
\end{multline*}
By the independence of the $Z_k$'s,~\eqref{eq:Ij} becomes
\begin{equation}\label{eq:plug3}
I(j_1,\dots,j_r) = \proba {\Big(} \vv Z \in \Omega_{H; j_1,\dots,j_r}(D_H) \cap E {\Big)}\,.
\end{equation}
By plugging~\eqref{eq:plug3} into~\eqref{eq:partition}, we obtain 
\begin{multline}\label{eq:intDH}
 \int_{D_H} \nu(\vv x, E) \proba^{\vv X}(\d \vv x)  = \int_{\calP_H(D_H)}\nu(\vv x,E)\Q_H^{\vv X_{\vv r}}(\d \vv x_{\vv r}) \\
 = \sum_{j_1\in J^{(1)}, \cdots, j_r\in J^{(r)}}\proba( \vv Z \in \Omega_{H; j_1,\cdots,j_r}(D_H) \cap E)  = \proba(A\odot \vv Z\in D_H, \vv Z\in E),
\end{multline}
because the summation over $(j_1,\dots,j_r)$ accounts for all relevant hitting scenarios corresponding to the matrix $H$.
Plugging~\eqref{eq:intDH} into~\eqref{eq:sumH}, we have
\[
\int_D\nu(\vv x,E)\proba^{\vv X}(\d \vv x) 
 = \sum_{H\in\calH(A)}\proba( \vv X \equiv A\odot \vv Z\in D_H,\ \vv Z \in  E) = \proba(\vv X\in D, \vv Z\in E)\,.
\]
This completes the proof of Theorem~\ref{thm:rcp}.
\end{proof}

%
%
\begin{proof}[Proof of Lemma \ref{lem:density}]
Consider the random vector $\vv X_{\vv r} = ( X_1,\dots,X_r)$. Observe that by the definition of the
set $\calP_H(D_H)$, on {\em the event } $\{ \vv X_{\vv r} \in \calP_H(D_H)\}$, we have
\begin{equation}\label{e:X_r-via-Z}
 \vv X_{\vv r} = \sum_{j_1 \in J^{(1)},\ \cdots,\ j_r \in J^{(r)}}
 \left (  \begin{array}{c} a_{1,j_1} Z_{j_1} \\ \vdots \\ a_{r,j_r} Z_{j_r} \end{array}\right) 
 \prod_{s=1}^r \underbrace{\ind{ {\Big\{} \bigvee_{k\in \wb J^{(s)} \setminus\{j_s\} } a_{s,k}Z_k < a_{s,j_s} Z_{j_s} {\Big\}}}}_{=:\ind{\{C_{s,j_s}\}}}.
\end{equation}
Note that since $J(s) \subset \wb J^{(s)},\ s = 1,\dots, r$, the events $\bigcap_{s=1}^r C_{s,j_s}$ are disjoint for all $r$-tuples
$(j_1,\dots,j_r)\in J^{(1)} \times \cdots \times J^{(r)}$.

Recall that our goal is to establish \eqref{e:lem:density}. By the fact that the sum in \eqref{e:X_r-via-Z} involves only one non-zero 
term for some $(j_1,\dots,j_r)$, with probability one, we have that for all measurable set  $\Delta \subset \calP_H(D_H)$, writing $\xi_{j_s} = a_{s,j_s} Z_{j_s}$,
\begin{multline}\label{e:Q-via-xi}
 \Q_H^{\vv X_{\vv r}} (\Delta) \equiv \proba( \vv X_{\vv r} \in \Delta ) \\
 = \sum_{j_1 \in J^{(1)},\ \cdots,\ j_r \in J^{(r)}}
 \proba\bpp{  \{ (\xi_{j_1},\cdots,\xi_{j_r}) \in \Delta\} \cap\bpp{ \bigcap_{s=1}^r C_{s,j_s}} }\,.
\end{multline}

Now, consider the last probability, for {fixed} $(j_1,\dots,j_r)$.  The random variables $\xi_{j_s},\ s = 1,\dots, r$
are independent and they have densities $f_{Z_{j_s}}(x_s/a_{s,j_s})/a_{s,j_s},\ x_s \in \mathbb R_+$.  We also have that the events
$C_{s,j_s},\ s = 1,\dots, r$ are mutually independent, since their definitions involve $Z_k$'s indexed by disjoint sets 
$\wb J^{(s)},\ s = 1,\dots, r$. By conditioning on the $\xi_{j_s}$'s, we obtain that the probability in the right-hand side 
of \eqref{e:Q-via-xi} equals
\begin{eqnarray*}
 & & \int_{\Delta} {\Big(}\prod_{s=1}^r \frac{1}{a_{s,j_s}} f(x_s/a_{s,j_s})  {\Big)} \times \prod_{s=1}^r 
 \proba {\Big(} \bigvee_{k\in\wb J^{(s)}\setminus\{j_s\} } a_{s,k}Z_k < x_s {\Big)}  \d x_1 \cdots \d x_r \\
 & & \ \ =\int_{\Delta}  \prod_{s=1}^r {\Big(} \frac{1}{a_{s,j_s}} f(x_s/ a_{s,j_s}) \prod_{k\in\wb J^{(s)}\setminus\{j_s\}} 
 F_{Z_k}(x_s/a_{s,k}) {\Big)}  \d x_1 \cdots \d x_r.
\end{eqnarray*}

In view of~\eqref{e:X_r-via-Z} and \eqref{eq:nuws}, replacing $\sum_{j_1\in J^{(1)},\ \cdots,\ j_r\in J^{(r)}} \prod_{s=1}^{r} \cdots$
by $\prod_{s=1}^{r} {(} \sum_{j \in J^{(s)}} \cdots {)}$, we obtain that the measure $\Q_H^{\vv X_{\vv r}}$ has a density on
$\calP(D_H)$, given by \eqref{e:lem:density}.\end{proof}
 

\appendix
\def\sm{{\sigma(m)}}
\def\wtsm{{\widetilde\sigma(m)}}
\def\pizlh{\calP_{l,H}\inv(\Lambda_H)}
\def\piwtzlh{\calP_{\wt l,H}\inv(\Lambda_H)}
\def\zlh{\Omega_l^H(\Lambda)}
\section{Regular conditional probability} \label{sec:rcp}

We recall here the notion of regular conditional probability. Let $\vv Z = (Z_1,\dots,Z_p)$, 
$\vv X = (X_1,\dots,X_n)$, and let $\calB_{\mathbb R^p_+}$ denote the Borel $\sigma$-algebra on $\mathbb R_+^p$.
The {\em regular conditional probability}  $\nu$ of $\vv Z$ given $\sigma(\vv X)$, is a function from 
$\calB_{\mathbb R^p_+}\times \mathbb R^n$ to $[0,1]$, such that 
\begin{itemize}
\item[(i)] $\nu(\vv x,\cdot)$ is a probability measure, for all $\vv x\in\mathbb R^n$,
\item[(ii)] The function $\nu(\cdot,E)$ is measurable, for all Borel sets $E\in {\cal B}_{\mathbb R^p}$.
\item[(iii)] $\P(\vv Z\in E,\ \vv X\in D) = \int_{D} \nu(\vv x,E) \proba^{\vv X}(\d \vv x)$, for all $E\in {\cal B}_{\mathbb R^p}$  and $D \in {\cal B}_{\mathbb R^n}$, where $\proba^{\vv X}(\cdot)\defe \P( \vv X\in \cdot)$.
\end{itemize}
See e.g.\ Proposition A 1.5.III in~\cite{daley08introduction} for more details. 

In Section~\ref{sec:conditionalProbability}, we provided an expression for the regular conditional 
probability in the max-linear model~\eqref{rep:maxLinear}: 
\equh\label{eq:nuEx}
\nu(\vv x,E) \defe \proba(\vv Z\in E\mid\vv X = \vv x),\ \  E\in\calB_{\mathbb R_+^p},\ \vv x\in\mathbb R^n_+\,.
\eque
The definition of $\nu$ implies that
$$
 \int_{\mathbb R^p} g(\vv z) \nu(\vv X,\d\vv z) = \esp (g(\vv Z) \mid\sigma(\vv X)),\ \ \proba^{\vv X}\mbox{-almost surely},
$$
for all Borel functions $g:\mathbb R^p \to \mathbb R$ with $\esp |g(\vv Z)| <\infty$.  By the strong law 
of large numbers, the latter conditional expectations are readily approximated by $N\inv \sum_{i=1}^N g(\vv Z^{(i)})$,
where $\vv Z^{(i)},\ i = 1,\dots, N$  are independent samples from the regular conditional probability 
$\nu(\vv X,\d \vv z)$.  Thus, $\nu$ is the {\em right} distribution to sample from when performing
prediction, given prior observed data. 


\end{document}

%% file: include_macros.tex
\newcommand{\ind}{{\bf 1}}
\newcommand{\proba}{\mathbb P}
\newcommand{\esp}{{\mathbb E}}

\newcommand{\defe}{\mathrel{\mathop:}=}

\newcommand{\inv}{^{-1}}
\newcommand{\argmin}{{\rm{argmin}}}
\newcommand{\argmax}{{\rm{argmax}}}

\newcommand{\calB}{{\cal B}}

\newcommand{\filF}{{\cal F}}

\newcommand{\calH}{{\cal H}}

\newcommand{\calJ}{{\cal J}}

\newcommand{\calN}{{\cal N}}

\newcommand{\calP}{{\cal P}}

\newcommand{\calR}{{\cal R}}

\def\indt#1{\{#1_t\}_{t\in T}}

\def\indz#1{\{#1_k\}_{k\in \mathbb Z}}

\newcommand{\eqnh}{\begin{eqnarray*}}
\newcommand{\eqne}{\end{eqnarray*}}
\newcommand{\eqnhn}{\begin{eqnarray}}
\newcommand{\eqnen}{\end{eqnarray}}
\newcommand{\equh}{\begin{equation}}
\newcommand{\eque}{\end{equation}}

\def\sif#1#2{\sum_{#1=#2}^\infty}

\def\bveein{\bigvee_{i=1}^n}

\def\bveejp{\bigvee_{j=1}^p}

\newcommand{\eqd}{\stackrel{\rm d}{=}}

\newcommand{\widebar}{\overline}
\newcommand{\eintt}{\ \int^{\!\!\!\!\!\!\!e}}
\newcommand{\Eintt}{\int^{\!\!\!\!\!\!\!e}}

\def\Eint#1{\, \int^{\!\!\!\!\!\!\!\!{e}}_{#1}}

\def\topp#1{^{(#1)}}

\def\bccbb#1{\Big\{#1\Big\}}
\def\sccbb#1{\{#1\}}
\def\bpp#1{\Big(#1\Big)}
\def\spp#1{(#1)}

\def\bb#1{\left[#1\right]}

\def\vv#1{{\bf #1}}

\def\d{{\rm d}}

\def\bbR{{\mathbb R}}

\def\P{{\mathbb P}}

\def\mand{\mbox{ and }}
\def\qmand{\quad\mbox{ and }\quad}

\def\mwith{\mbox{ with }}
\def\qmwith{\quad\mbox{ with }\quad}
\def\mfa{\mbox{ for all }}

\def\adaptF#1{\{#1_t,\filF_t:0\leq t<\infty\}}

\def\wt#1{\widetilde{#1}}
\def\wb#1{\widebar{#1}}
\def\wht#1{\widehat{#1}}

\def\ifhead#1#2{\left\{\begin{array}{#1@{\quad\mbox{ if }\quad}#2}}
\def\ifend{\end{array}\right.}

%% file: CSTheory_stdFormat1.bbl
\begin{thebibliography}{10}

\bibitem{balkema77max}
A.~A. Balkema and S.~I. Resnick.
\newblock Max-infinite divisibility.
\newblock {\em J. Appl. Probability}, 14(2):309--319, 1977.

\bibitem{buishand08spatial}
T.~Buishand, L.~de~Haan, and C.~Zhou.
\newblock On spatial extremes: With application to a rainfall problem.
\newblock {\em Ann. Appl. Stat.}, 2(2):624--642, 2008.

\bibitem{caprara00algorithms}
A.~Caprara, P.~Toth, and M.~Fischetti.
\newblock Algorithms for the set covering problem.
\newblock {\em Ann. Oper. Res.}, 98:353--371 (2001), 2000.
\newblock Optimization theory and its application (Perth, 1998).

\bibitem{cooley07bayesian}
D.~Cooley, D.~Nychka, and P.~Naveau.
\newblock Bayesian spatial modeling of extreme precipitation return levels.
\newblock {\em J. Amer. Statist. Assoc.}, 102(479):824--840, 2007.

\bibitem{daley08introduction}
D.~J. Daley and D.~Vere-Jones.
\newblock {\em An introduction to the theory of point processes. {V}ol. {II}}.
\newblock Probability and its Applications (New York). Springer, New York,
  second edition, 2008.
\newblock General theory and structure.

\bibitem{davis89basic}
R.~A. Davis and S.~I. Resnick.
\newblock Basic properties and prediction of max-{ARMA} processes.
\newblock {\em Adv. in Appl. Probab.}, 21(4):781--803, 1989.

\bibitem{davis93prediction}
R.~A. Davis and S.~I. Resnick.
\newblock Prediction of stationary max-stable processes.
\newblock {\em Ann. Appl. Probab.}, 3(2):497--525, 1993.

\bibitem{davison90models}
A.~C. Davison and R.~L. Smith.
\newblock Models for exceedances over high thresholds.
\newblock {\em J. Roy. Statist. Soc. Ser. B}, 52(3):393--442, 1990.
\newblock With discussion and a reply by the authors.

\bibitem{dehaan78characterization}
L.~de~Haan.
\newblock A characterization of multidimensional extreme-value distributions.
\newblock {\em Sankhy\=a Ser. A}, 40(1):85--88, 1978.

\bibitem{dehaan84spectral}
L.~de~Haan.
\newblock A spectral representation for max-stable processes.
\newblock {\em Ann. Probab.}, 12(4):1194--1204, 1984.

\bibitem{dehaan06extreme}
L.~de~Haan and A.~Ferreira.
\newblock {\em Extreme value theory}.
\newblock Springer Series in Operations Research and Financial Engineering.
  Springer, New York, 2006.
\newblock An introduction.

\bibitem{dehaan06spatial}
L.~de~Haan and T.~T. Pereira.
\newblock Spatial extremes: Models for the stationary case.
\newblock {\em The Annals of Statistics}, 34:146--168, 2006.

\bibitem{dehaan86stationary}
L.~de~Haan and J.~Pickands, III.
\newblock Stationary min-stable stochastic processes.
\newblock {\em Probab. Theory Relat. Fields}, 72(4):477--492, 1986.

\bibitem{furrer09fields}
R.~Furrer, D.~Nychka, and S.~Sain.
\newblock {\em fields: Tools for spatial data}, 2009.
\newblock R package version 6.01.

\bibitem{gine90max}
E.~Gin{\'e}, M.~G. Hahn, and P.~Vatan.
\newblock Max-infinitely divisible and max-stable sample continuous processes.
\newblock {\em Probab. Theory Related Fields}, 87(2):139--165, 1990.

\bibitem{kabluchko09stationary}
Z.~Kabluchko, M.~Schlather, and L.~de~Haan.
\newblock Stationary max-stable fields associated to negative definite
  functions.
\newblock {\em Ann. Probab.}, 37(5):2042--2065, 2009.

\bibitem{naveau09modelling}
P.~Naveau, A.~Guillou, D.~Cooley, and J.~Diebolt.
\newblock Modelling pairwise dependence of maxima in space.
\newblock {\em Biometrika}, 96(1):1--17, 2009.

\bibitem{R09R}
{R Development Core Team}.
\newblock {\em R: A Language and Environment for Statistical Computing}.
\newblock R Foundation for Statistical Computing, Vienna, Austria, 2009.
\newblock {ISBN} 3-900051-07-0.

\bibitem{resnick87extreme}
S.~I. Resnick.
\newblock {\em Extreme values, regular variation, and point processes},
  volume~4 of {\em Applied Probability. A Series of the Applied Probability
  Trust}.
\newblock Springer-Verlag, New York, 1987.

\bibitem{resnick07heavy}
S.~I. Resnick.
\newblock {\em Heavy-tail phenomena}.
\newblock Springer Series in Operations Research and Financial Engineering.
  Springer, New York, 2007.
\newblock Probabilistic and statistical modeling.

\bibitem{resnick91random}
S.~I. Resnick and R.~Roy.
\newblock Random usc functions, max-stable processes and continuous choice.
\newblock {\em Ann. Appl. Probab.}, 1(2):267--292, 1991.

\bibitem{schlather02models}
M.~Schlather.
\newblock Models for stationary max--stable random fields.
\newblock {\em Extremes}, 5:33--44, 2002.

\bibitem{schlather03dependence}
M.~Schlather and J.~A. Tawn.
\newblock A dependence measure for multivariate and spatial extreme values:
  Properties and inference.
\newblock {\em Biometrika}, 90:139--156, 2003.

\bibitem{smith90max}
R.~L. Smith.
\newblock Max--stable processes and spatial extremes.
\newblock unpublished manuscript, 1990.

\bibitem{srivastava98course}
S.~M. Srivastava.
\newblock {\em A course on {B}orel sets}, volume 180 of {\em Graduate Texts in
  Mathematics}.
\newblock Springer-Verlag, New York, 1998.

\bibitem{stoev06extremal}
S.~A. Stoev and M.~S. Taqqu.
\newblock Extremal stochastic integrals: a parallel between max-stable and
  alpha-stable processes.
\newblock {\em Extremes}, 8(3):237--266, 2006.

\bibitem{wang10maxLinear}
Y.~Wang.
\newblock {\em maxLinear: Conditional sampling for max-linear models}, 2010.
\newblock R package version 1.0.

\bibitem{wang10structure}
Y.~Wang and S.~A. Stoev.
\newblock On the structure and representations of max--stable processes.
\newblock {\em Adv. in Appl. Probab.}, 42(3):855--877, 2010.

\end{thebibliography}
